\pdfoutput=1
\documentclass[11pt]{article}
\usepackage[margin=1in]{geometry}

\usepackage[ruled]{algorithm2e} %

\SetAlFnt{\small}
\SetAlCapFnt{\small}
\SetAlCapNameFnt{\small}
\SetAlCapHSkip{0pt}
\IncMargin{-\parindent}

\usepackage{color}
\usepackage{amsmath,amsthm,amsfonts,amssymb}
\usepackage{float}
\usepackage{graphicx}
\usepackage[inline]{enumitem}
\usepackage{bm}
\usepackage{tikz,pgfplots}
\pgfplotsset{compat=1.7}
\usepackage{subcaption}
\usepackage[numbers]{natbib}

\renewcommand{\vec}{\bm}
\renewcommand{\th}{{\text{th}}}
\newcommand{\erm}{L_1\textrm{-ERM}}
\newcommand{\conv}{\mathrm{Conv}}

\SetArgSty{textrm}  %
\SetKwProg{Fn}{Function}{}{end}

\newtheorem{theorem}{Theorem}
\newtheorem{lemma}[theorem]{Lemma}
\newtheorem{proposition}[theorem]{Proposition}
\theoremstyle{definition}
\newtheorem{definition}{Definition}

\newtheorem{example}{Example}

\newcount\Comments  %
\newcommand{\kibitz}[2]{\ifnum\Comments=1{\color{#1}{#2}}\fi}

\usepackage[]{color-edits}
\addauthor{yc}{green}
\addauthor{cp}{red}
\addauthor{ns}{blue}
\addauthor{ap}{cyan}

\renewcommand{\tilde}{\widetilde}
\newcommand{\set}[1]{\{#1\}}
\newcommand{\eps}{\epsilon}

\newcommand{\calD}{\mathcal{D}}
\newcommand{\calF}{\mathcal{F}}
\newcommand{\calP}{\mathcal{P}}
\newcommand{\calS}{\mathcal{S}}
\newcommand{\bbR}{\mathbb{R}}

\newcommand{\bbRbar}{\overline{\bbR}}
\newcommand{\bbN}{\mathbb{N}}
\newcommand{\sign}{\textrm{sign}}
\DeclareMathOperator*{\med}{med} %
\DeclareMathOperator*{\argmin}{arg\,min}

\renewcommand{\hat}{\widehat}

\newcommand{\tempcom}{\vec{x}}
\makeatletter
\def\vx{\@ifnextchar(\@f\tempcom}
\def\@f(#1){\vec{x}_{#1}}
\makeatother

\newcommand{\btempcom}{\bar{\vec{x}}}
\makeatletter
\def\bvx{\@ifnextchar(\@fbar\btempcom}
\def\@fbar(#1){\bar{\vec{x}_{#1}}}
\makeatother

\renewcommand{\bar}{\overline}
\newcommand{\vk}{\vec{k}}
\newcommand{\vy}{\vec{y}}
\newcommand{\vbeta}{\vec{\beta}}
\newcommand{\vbetaone}{\vec{\beta}_1}
\newcommand{\ty}{\tilde{y}}
\newcommand{\hy}{\hat{y}}
\newcommand{\tf}{\tilde{f}}
\newcommand{\hf}{\hat{f}}
\newcommand{\tS}{\tilde{S}}
\newcommand{\hS}{\hat{S}}
\newcommand{\hR}{\hat{R}}
\newcommand{\tvy}{\vec{\tilde{y}}}
\newcommand{\CWA}{\text{CWA}}
\newcommand{\DA}{\text{DA}}

\newcommand{\pref}{\succcurlyeq}
\newcommand{\stpref}{\succ}
\newcommand{\lbr}{\left\langle}
\newcommand{\rbr}{\right\rangle}
\DeclareMathOperator{\rss}{RSS}

\begin{document}
\title{Strategyproof Linear Regression in High Dimensions\thanks{A preliminary version of this paper was published in the proceedings of the 19th ACM Conference on Economics and Computation (EC), 2018. This work was partially supported by the National Science Foundation under grants CCF-1718549, IIS-1350598, IIS-1714140, CCF-1525932, and CCF-1733556; by the Office of Naval Research under grants N00014-16-1-3075 and N00014-17-1-2428; by a Sloan Research Fellowship and a Guggenheim Fellowship; and by the Natural Sciences and Engineering Research Council of Canada (NSERC) under the Discovery Grants program.}}
\author{Yiling Chen\thanks{Harvard University, Cambridge, MA 02138, USA. {\tt yiling@seas.harvard.edu}.}
	\and Chara Podimata\thanks{Harvard University, Cambridge, MA 02138, USA. {\tt podimata@g.harvard.edu}.}
	\and Ariel D. Procaccia\thanks{Carnegie Mellon University, Pittsburgh, PA 15213, USA. {\tt arielpro@cs.cmu.edu}.}
	\and Nisarg Shah\thanks{University of Toronto, Toronto, ON M5S 3G8, Canada. {\tt nisarg@cs.toronto.edu}.}}
\maketitle
\begin{abstract}
	This paper is part of an emerging line of work at the intersection of machine learning and mechanism design, which aims to avoid noise in training data by correctly aligning the incentives of data sources. Specifically, we focus on the ubiquitous problem of \emph{linear regression}, where \emph{strategyproof} mechanisms have previously been identified in two dimensions. In our setting, agents have single-peaked preferences and can manipulate only their response variables. Our main contribution is the discovery of a family of \emph{group strategyproof} linear regression mechanisms in any number of dimensions, which we call \emph{generalized resistant hyperplane} mechanisms. The game-theoretic properties of these mechanisms --- and, in fact, their very existence --- are established through a connection to a discrete version of the Ham Sandwich Theorem. 
\end{abstract}

\section{Introduction}

Designing machine learning algorithms that are robust to noise in training data is a topic of intense research. A large body of work addresses stochastic noise~\cite{L91,GS95}. On the other extreme, another branch of the literature focuses on adversarial noise~\cite{KL93,BEK02,CCM13}, that is, errors are introduced by an adversary with the explicit purpose of sabotaging the algorithm. The latter approach is often too pessimistic, and generally leads to negative results. 

More recently, some researchers have taken a game-theoretic viewpoint; it suggests a model of \emph{strategic noise} that can be seen as occupying the middle ground of noise models. Specifically, training data is provided by strategic sources --- hereinafter \emph{agents} --- that may intentionally introduce errors \emph{to maximize their own benefit}. Compared to adversarial noise, the advantage of this model (when its underlying assumptions hold true) is that, if we aligned the agents' incentives correctly, it would be possible to obtain uncontaminated data. From this viewpoint, the ideal is the design of learning algorithms that in addition to being statistically efficient, are \emph{strategyproof}, i.e., where supplying pristine data is a dominant strategy for each agent. 

We subscribe to this agenda, and advance it in the context of the ubiquitous problem of linear regression, i.e., fitting a hyperplane through given data. We consider agents who can manipulate their dependent variables in order to minimize their vertical distance from the output hyperplane, and design strategyproof regression mechanisms without payments. 

When does this type of strategic regression problem arise? \citet{DFP10} give the real-world example of the global fashion chain Zara, whose distribution process relies on regression~\cite{CG10}. Specifically, the demand for each product at each store is predicted based on historical data, as well as information provided by store managers. Since the supply of popular items is limited, store managers may strategically manipulate requested quantities so that the output of the regression process would better fit their needs, and, indeed, there is ample evidence that many of them have done so~\cite{CGMT+10}. More generally, as discussed in detail by \citet{PP04}, this type of setting is relevant whenever ``data could come from surveys composed by agents interested in not being perceived as real outliers if the estimation results could be used in the future to change the economic situation of the agents that generate the sample.'' 

\subsection{Our Model and Results}

A bit more formally, we study a linear regression setting in which the task is to fit a hyperplane through data points $(\vx(i),y_i)$ for $i \in \set{1,\ldots,n}$, where $\vx(i) \in \bbR^d$ are the independent variables and $y_i \in \bbR$ is the dependent variable. Following \citet{DFP10} and \citet{PP04}, we assume that the independent variables are public information, but dependent variable $y_i$ is held privately by agent $i$. A mechanism elicits the private information of the agents, and returns a hyperplane represented by vector $\vbeta = (\vbetaone,\beta_0) \in \bbR^{d+1}$. Under this outcome, the residual for agent $i$ is $r_i = y_i - \vbetaone^T \vx(i) - \beta_0$, and, loosely speaking, agents wish to minimize $|r_i|$ (see Section~\ref{sec:model} for a precise description of agent preferences).

Our starting point is the work of \citet{DFP10}, who show that empirical risk minimization (ERM) with the $L_1$ loss (in short, $\erm$), coupled with a specific tie-breaking rule, is group strategyproof, that is, no coalition of agents can be weakly better off by misreporting. We extend this result and show that replacing the $L_1$ loss by a weighted $L_1$ loss and adding convex regularization to the risk function preserves group strategyproofness. But this still gives a relatively restricted family of strategyproof mechanisms, and we seek a broader understanding of what is possible in our setting. 

To that end, we look to the work of \citet{PP04}, who focus on the two-dimensional case (known as \emph{simple linear regression}), i.e., fitting a line through points on a plane. They propose a wide family of strategyproof mechanisms, which they call {\em clockwise repeated median} (CRM) mechanisms. These mechanisms are parametrized by two subsets of agents $S$ and $S'$. \citet{PP04} establish conditions on $S$ and $S'$ under which they claim that CRM mechanisms are strategyproof. We identify a mistake in this result, present counterexamples showing violation of strategyproofness under their conditions, and identify three stricter conditions under which we can recover strategyproofness --- in fact, we prove group strategyproofness. Under one of our conditions, CRM mechanisms coincide with a family of mechanisms from the statistics literature known as \emph{resistant line mechanisms}~\cite{JV85}. Our work therefore establishes the group strategyproofness of these mechanisms.

Our main result is that we generalize the CRM family to higher dimensions, thereby justifying the title of this paper. We introduce the family of {\em generalized resistant hyperplane} (GRH) mechanisms, which, to the best of our knowledge, is the first extension of resistant line mechanisms beyond the plane. In $d+1$ dimensions, GRH mechanisms are parametrized by $d+1$ subsets of agents. Through a surprising connection to the literature on the Ham Sandwich Theorem, we find a condition on the subsets under which GRH mechanisms are group strategyproof. Strikingly, our proof of this general \emph{group strategyproofness} result in \emph{any number} of dimensions is much shorter than the (incorrect) proof of \citet{PP04} for the \emph{strategyproofness} of CRM mechanisms in \emph{two dimensions}.

We also study a property called impartiality, which is stricter than strategyproofness. We establish the existence of a wide family of impartial mechanisms, which, unlike our generalized $\erm$ and generalized resistant hyperplane mechanisms, are strategyproof but not group strategyproof (except for constant functions). Building upon the work of~\citet{Moul80}, we also provide two non-constructive characterizations of strategyproof mechanisms for linear regression. 

Strategyproofness is not the sole desideratum; constant functions (e.g., the flat hyperplane $y=0$) are strategyproof but not necessarily desirable. We would also like the mechanism to have good statistical efficiency. For that, we compare (families of) strategyproof mechanisms in terms of their approximation of the optimal squared loss, leveraging our characterization.  %
Most importantly, we establish a lower bound of $2$ on the approximation ratio of any strategyproof mechanism, which means that any mechanism that is even close to \emph{ordinary least squares} regression must be manipulable.

\subsection{Related Work}

As discussed above, our work is most closely related to that of \citet{PP04} and \citet{DFP10}. Here we try to give a broader picture of the state of research on machine learning algorithms that are robust to strategic noise. This research can be categorized using three key axes: (i) manipulable information, (ii) goal of the agents, and (iii) use of payments and incentive guarantees. 

On the first axis, like us, most papers assume that independent variables (or \emph{feature vectors} in the language of classification) are public information, and dependent variables (labels) are private, manipulable information~\cite{DFP10,MPR12, PP03, PP04}, though some papers also design algorithms robust to strategic feature vectors~\cite{HMPW16, DRZWW17}. 
\citet{MPR12} provide strong positive results for designing strategyproof classifiers when there are either only two classifiers, or the agents are interested in a shared set of input points. On the other hand, \citet{HMPW16} study the problem of constructing classifiers that are robust to agents strategically misreporting their \emph{feature vector}, in order to trick the algorithm into misclassifying them. Their setting is modeled as a one-shot Stackelberg game. The more recent work of \citet{DRZWW17} models the same problem in an online setting; they provide guarantees that ensure that the problem is convex, and, therefore, they are able to derive a computationally efficient learning algorithm that has diminishing \emph{Stackelberg regret}.

On the second axis, one line of research focuses on agents motivated by privacy concerns, with a tradeoff between accuracy and privacy~\cite{CIL15,CDP15}; another focuses on agents who want the algorithm to make accurate assessment on their own sample, even if this reduces the overall accuracy. This form of strategic manipulation has been studied for estimation~\cite{CPS16b}, classification~\cite{MPR10, MAMR11, MPR12}, and regression~\cite{PP04, DFP10} problems. Our problem falls squarely into the second category.  

Finally, on the third axis, various papers differ on whether monetary payments to agents are allowed~\cite{CDP15}, and on how strongly to guarantee truthful reporting: the stronger strategyproofness requirement~\cite{PP03,PP04,MPR12} versus the weaker Bayes-Nash incentive compatibility~\cite{IL13, CIL15}. Our work falls into the literature of mechanism design without money; we study linear regression mechanisms that enforce strategyproofness without paying the agents, or asking the agents to pay.

\section{Model}
\label{sec:model}

Let $[k] \triangleq \{1,\ldots,k\}$ be the set of first $k$ natural numbers, and $\bbRbar = \bbR \cup \set{-\infty,\infty}$ be the extended real line. Given numbers $t_1,\ldots,t_k \in \bbRbar$, let $\min(t_1,\ldots,t_k)$ denote the smallest value, and $\min^j(t_1,\ldots,t_k)$ denote the $j^\th$ smallest value. Let $\med(t_1,\ldots,t_k)$ denote their median: when $k$ is odd, this is equal to $\min^{(k+1)/2}(t_1,\ldots,t_k)$, but when $k$ is even, this could be either $\min^{k/2}(t_1,\ldots,t_k)$ (the ``left median'') or $\min^{k/2+1}(t_1,\ldots,t_k)$ (the ``right median'').\footnote{This is different from the standard definition, which takes the average of the left and right medians, but necessary to ensure incentive guarantees.} 

Our work focuses on the problem of linear regression, i.e., fitting a hyperplane through given data. Let $N = [n]$. We are given a collection of data points $\calD = (\vx(i),y_i)_{i \in N}$, where $\vx(i) \in \bbR^d$ and $y_i \in \bbR$ are called the \emph{independent} and \emph{dependent} variables of point $i$, respectively. Let $\bvx(i) = (\vx(i),1)$. Our goal is to find a vector $\vbeta = (\vbeta_1,\beta_0) \in \bbR^{d+1}$ such that $\vbeta^T \bvx(i) = \vbeta_1^T \vx(i) + \beta_0$ is a good approximation of $y_i$ for each $i \in N$. The quantity $r_i = y_i-\vbeta^T \bar{\vx(i)}$ is called the residual of point $i$. %

\paragraph{Strategic setting.} We study a setting in which each data point $p_i = (\vx(i),y_i)$ is provided by a strategic agent $i$. We also denote the set of agents by $N$. Following \citet{PP04} and \citet{DFP10}, we assume that the independent variables $\vx = (\vx(i))_{i \in N}$ constitute \emph{public} information, which the agents cannot manipulate. Each agent $i$ holds the dependent variable $y_i$ as private information, and may report a different value $\ty_i$ in order to receive a more preferred outcome. Thus, the principal observes the reported data points $\tilde{\calD} = (\vx(i),\ty_i)_{i \in N}$. Let us denote $\vy = (y_i)_{i \in N}$ and $\tvy = (\ty_i)_{i \in N}$. 

\paragraph{Mechanisms.} Because the agents cannot change $\vx$, we can effectively treat it as fixed. A mechanism for linear regression $M^{\vx}$ is therefore defined for given public information $\vx$, takes as input reported private information $\tvy$, and returns a vector $\vbeta$. %
We omit $\vx$ when it is clear from the context.

\paragraph{Agent preferences.} When a mechanism returns $\vbeta$, we say that the outcome for agent $i$ is $\hy_i(\vbeta) = \vbeta^T\ \vx(i)$. We omit $\vbeta$ when it is clear from the context. The agent only cares about her own outcome $\hy_i$, and would like it to be as close to $y_i$ as possible. Formally, we assume that agent $i$ has \emph{single-peaked preferences}~\cite{Black58,Moul80} over $\hy_i$ with peak at $y_i$. We represent the weak preference relation by $\pref_i$ and the strict preference relation by $\stpref_i$. Formally, for all $a,b \in \bbR$, $y_i > a \ge b$ or $y_i < a \le b$ must imply $y_i \stpref_i a \pref_i b$. %

\paragraph{Game-theoretic desiderata.} Our goal is to prevent agents from misreporting their private information. The game theory literature offers a strong desideratum under which agents have no incentive to misreport even if they have know what the other agents would report.

\begin{definition}[Strategyproofness]
A mechanism $M^{\vx}$ is called \emph{strategyproof} (SP) if each agent weakly prefers truthfully reporting her private information to misreporting it, regardless of the reports of the other agents. Formally, for each $i \in N$, $y_i \in \bbR$, and $\tvy \in \bbR^n$, we need $\hat{y_i}(M^{\vx}(y_i,\tvy_{-i})) \pref_i \hat{y_i}(M^{\vx}(\tvy))$. Note that this must hold for any possible single-peaked preferences the agent may have.
\end{definition}

While no individual agent can benefit from misreporting under a strategyproof mechanism, a group of agents may still be able to collude, and benefit by simultaneously misreporting. This can be prevented by imposing a stronger desideratum.

\begin{definition}[Group Strategyproofness]
A mechanism $M^{\vx}$ is called \emph{group strategyproof} (GSP) if no coalition of agents can simultaneously misreport in a way that no agent in the coalition is strictly worse off and some agent in the coalition is strictly better off, irrespective of the reports of the other agents. Formally, for each $S \subseteq N$, $\vec{y_S} = (y_i)_{i \in S} \in \bbR^{|S|}$, and $\tvy \in \bbR^n$, it should not be the case that $\hy_i(M^{\vx}(\tvy)) \pref_i \hy_i(M^{\vx}(\vy_S,\tvy_{N\setminus S}))$ for every $i \in S$, and %
the preference is strict for at least one $i \in S$.
\end{definition}

The game theory literature also considers a weaker notion of group strategyproofness in which not all the agents in a manipulating coalition should be strictly better off. We do not consider this notion because our group strategyproof mechanisms are able to satisfy the stronger notion. 

Note that we do not assume that the data points are generated by an underlying statistical process. Our results are independent of how the data points were generated.

\section{Families of Strategyproof Mechanisms}
\label{sec:spmechs}

In this section, we analyze families of (group) strategyproof mechanisms for linear regression. Our results generalize existing families of mechanisms, and propose novel families. 

\subsection{Empirical Risk Minimization with the $L_1$ Loss}
\label{subsec:l1-ERM}

Consider a single dimensional setting, in which each agent $i$ has a private value $y_i$, reports a possibly different value $\ty_i$, and the mechanism returns a single value $\hy$. Each agent $i$ has single-peaked preferences over $\hy$ with peak at $y_i$. This corresponds to the special case of our setting in which $\vx(i) = \vx(j)$ for all $i,j \in N$, or alternatively, the dimension $d=0$. In this setting, it has long been known that choosing the {\em median} of the reported values achieves group strategyproofness~\cite{DF61}. %
It can be shown that the median minimizes the sum of absolute ($L_1$) losses with respect to the reports, i.e., given $\vy$, it chooses $\argmin_{y \in \bbR} \sum_{i=1}^n |y-y_i|$, with an appropriate tie-breaking when $n$ is even. In the machine learning terminology, the median is the empirical risk minimizer (ERM) with the $L_1$ loss. 

Inspired by this, \citet{DFP10} study ERM with the $L_1$ loss in a more general regression setting, and show that it remains group strategyproof. Specifically, they focus on finding a (potentially non-linear) regression function $f: \bbR^d \to \bbR$ from a given convex set $\calF$. %
Given $\calD = (\vx(i),y_i)_{i \in N}$, define the empirical $L_1$ risk of a regression function $f \in \calF$ as $\hat{R}(f,\calD) = \sum_{i \in N} |y_i-f(\vx(i))|$. Let $\|\cdot\| : \calF \to \bbR$ be a strictly convex function. They show that minimizing the empirical $L_1$ risk, and breaking ties among the optimal solutions by minimizing $\|\cdot\|$ is group strategyproof. We refer to this mechanism by $\erm$ \footnote{For a formal description of the algorithm, we refer the interested reader to the full version of our paper.}. For linear regression, this approach is known by various names in the literature, such as Least Absolute Deviations (LAD), Minimum Sum of Absolute Errors (MSAE), or Least Absolute Value (LAV). The tie-breaking step is crucially required because the empirical $L_1$ risk may have multiple minimizers.

We present a generalization of their mechanism while retaining group strategyproofness. In particular, we extend the objective function $\hat{R}$ in two ways: i) we allow a weighted $L_1$ loss, in which the loss of each agent $i$ is multiplied by a weight $w_i^{\vx}$, and ii) we allow adding a convex regularizer $h : \calF \to \bbR$. Note that regularization is widely used in machine learning to prevent ERM from overfitting. Our generalization, which we term {\em generalized $\erm$}, is presented as Algorithm~\ref{alg:gen-l1}. While we are only interested in linear regression, we note that generalized $\erm$ works for the general regression setting of~\citet{DFP10}. 

\begin{algorithm}[ht]
	\SetAlgoLined
	\KwIn{Data points $\calD = (\vx(i),y_i)_{i \in N}$, convex hypothesis space $\calF$, constants $(w_i^{\vx})_{i \in N}$, convex regularizer $h : \calF \to \bbR$, strictly convex function $\|\cdot\|: \calF \to \bbR$.}
	\KwOut{Function $f^* \in \calF$.}
	\vskip .5em
	$\forall f \in \calF,\ \hat{R}(f,\calD) \triangleq \sum_{i \in N} w_i^{\vx} \cdot |y_i-f(\vx(i))| + h(f)$\;
	$r^* \gets \inf_{f \in \calF}\; \hat{R}(f,\calD)$\;
	\Return $f^* \gets \argmin_{f \in \calF : \hat{R}(f,\calD) = r^*}\; \|f\|$\;
	\caption{Generalized $\erm$ (Regularized ERM with a weighted $L_1$ loss)}
	\label{alg:gen-l1}
\end{algorithm}

\begin{theorem}\label{thm:gen-l1-gsp} 
	Generalized $\erm$ is a group strategyproof regression mechanism.
\end{theorem}

Our proof, presented in Appendix~\ref{appendix:missing-proofs} for completeness, essentially mirrors the proof of \citet{DFP10}; we identify three steps in their proof where they use the structure of the risk function $\hat{R}$, and observe that these steps follow through with our more general risk function. 

There are several potential advantages of generalized $\erm$ over the vanilla $\erm$. First, generalized $\erm$ allows eliminating the tie-breaking step if the new risk function is guaranteed to have a unique minimizer. For instance, adding a {\em strictly convex} regularizer would achieve this. 

Second, for the aforementioned single dimensional setting, \citet{Moul80} proved that every strategyproof\footnote{\citet{Moul80} shows that for the single dimensional setting, strategyproofness is equivalent to group strategyproofness.} and anonymous\footnote{A mechanism is anonymous if permuting the reports of the agents does not change the output of the mechanism. This is a reasonable desideratum in the single dimensional setting due to the absence of public information that distinguishes agents naturally.} mechanism is a {\em generalized median}: for every $\alpha_1,\ldots,\alpha_{n+1} \in \bbRbar$, the corresponding generalized median returns $\med\set{y_1,\ldots,y_n,\alpha_1,\ldots,\alpha_{n+1}}$. Here, $\set{\alpha_j}_{j \in [n+1]}$ are called ``phantoms''. We can alternatively view this as returning $\argmin_{y \in \bbR} \sum_{i \in [n]} |y-y_i| + h(y)$, where $h(y) = \sum_{j \in [n+1] \text{ s.t. } \alpha_j \in \bbR} |y-\alpha_j| + (k_{-\infty}-k_{\infty}) \cdot y$, and for $t \in \set{-\infty,\infty}$, $k_t = |\set{j : \alpha_j = t}|$.\footnote{When all phantoms are finite, $h(y) = \sum_{j \in [n+1]} |y-\alpha_j|$. The term $|y-\alpha_j|$ has derivative $1$ when $y > \alpha_j$, and $-1$ when $y < \alpha_j$. For $\alpha_j = -\infty$ (resp. $\infty$), we can mimic this effect by adding a different term whose derivative is always $-1$ (resp. $1$).} Since $h(y)$ is a convex function, we can view it as a regularizer in our generalized $\erm$. Hence, for the single dimensional setting, generalized $\erm$ covers all generalized medians. In contrast, $\erm$ reduces to a specific mechanism in this family, the median. %

Finally, algorithms that add convex regularization to $\erm$ have been studied in the machine learning literature~\cite{WGZ06,Wang13}; our generalization establishes group strategyproofness of these algorithms.

We also note that in the statistics literature, the vanilla $\erm$ is treated as a member of the more general family of {\em quantile regression} mechanisms~\cite{KB79}, which, given $q \in [0,1]$, minimize the following empirical risk function: 
\begin{equation}
\hat{R}_q(f,\calD) = \sum_{i \in N: y_i \geq f(\vx(i))} q \cdot |y_i - f(\vx(i))| + \sum_{i \in N : y_i < f(\vx(i))} (1-q) \cdot |y_i - f(\vx(i))|.
\end{equation}
$\erm$ corresponds to the choice of $q=0.5$. In the one-dimensional setting, other values of $q$ correspond to different quantiles (i.e., correspond to $\min^k$ for various $k$), and thus induce strategyproof mechanisms. One might wonder if quantile regression remains strategyproof in higher dimensions. We answer this \emph{negatively} by providing an example in Appendix~\ref{appendix:qERM-cnt}, in which the quantile regression mechanism for $q=0.4$ is shown to violate strategyproofness. It is an interesting question to discover a strategyproof version of quantiles for linear regression.

\subsection{Generalized Resistant Hyperplane Mechanisms}
\label{sec:CRM}

In this section, we introduce a novel family of strategyproof mechanisms for linear regression. Our family extends the known family of resistant line mechanisms from the statistics literature~\cite{JV85}, which were only defined for simple linear regression ($d=1$), to higher dimensions. We first take a slight detour through a previous approach in the literature.

\subsubsection{A Detour Through Clockwise Repeated Median Mechanisms}

\citet{PP04} introduced a novel family of mechanisms, which they termed \emph{Clockwise Repeated Median} (CRM) mechanisms. CRM mechanisms are only defined for the special case of \emph{simple linear regression}, i.e., for fitting a straight line through a set of points on a plane. In describing these mechanisms, we use scalar notations where possible. For instance, we use $x_i$ to denote the x-coordinate of agent $i$, and $\beta_1$ to denote the slope of the regression line. For CRM mechanisms to be well defined, we also need to assume that the set of points is ``admissible''.

\begin{definition}[Admissible Set]
	A collection of data points $\mathcal{D} = (x_i, y_i)_{i \in N}$ is called \emph{admissible} if $x_i \neq x_j$ for all distinct $i,j \in N$.
\end{definition}

The CRM family is parametrized by two subsets of agents, $S,S' \subseteq N$. These subsets must be chosen based on the public information $\vx$, and therefore can be treated as fixed. Informally, given $S,S' \subseteq N$, the $(S,S')$-CRM mechanism first computes the median \emph{clockwise angle} (CWA), defined below, from each point $i \in S$ to points in $S'$. Then, it chooses the point $i^* \in S$ whose median CWA is the median of the median CWAs from all points in $S$. If the median CWA from point $i^*$ is towards point $j^* \in S'$, then the mechanism returns the straight line passing through points $i^*$ and $j^*$. Formally, the mechanism is defined as follows. \citet{PP04} established the equivalence of this formal definition and the aforementioned informal description.

\begin{definition}[CRM Mechanisms]
	Define the {\em clockwise angle} (CWA) from $(x_i,y_i)$ to $(x_j,y_j)$ as:
	\begin{equation}\label{eq:cwa}
	\CWA((x_i, y_i), (x_j,y_j)) = \pi + \sign(x_j - x_i)\cdot \frac{\pi}{2} + \sign \left( \frac{y_j - y_i}{x_j - x_i}\right) \left| \arctan \left(\frac{y_j - y_i}{x_j - x_i} \right) \right|.
	\end{equation}
	Given $\calD = (x_i,y_i)_{i \in N}$ and $S,S' \subseteq N$, let the {\em directing angle} be defined as:
	\begin{equation}\label{eq:DA}
	\DA(S,S') = \med_{i \in S} \med_{j \in S' : j \neq i} \CWA((x_i, y_i),(x_j, y_j)).
	\end{equation}
	Then, the $(S,S')$-CRM mechanism returns the line $\vbeta = (\beta_1,\beta_0)$ given by:
	\begin{align}
	\begin{split}
	\beta_1 &= \tan \left[ \DA(S,S') -\pi - \frac{\pi}{2}\cdot\sign\left( \DA(S,S') - \pi \right) \right],\\
	\beta_0 &= \med_{i \in S}\ (y_i - \beta_1\cdot x_i).
	\end{split}
	\label{eq:CRM}
	\end{align}	
\end{definition}

First, we notice that the definition of the CRM family uses three medians: two to define the directing angle $DA(S,S')$, and one to define the $y$-intercept $\beta_0$. Each median, when taken over an even number of values, can be the left median or the right median. While \citet{PP04} do not mention how these choices should be made, it is easy to check that in order to achieve the desired incentive properties, these choices cannot be made independently of each other. Later, we present a generalization which captures the different feasible choices in a simpler form.

\citet{PP04} claimed that the $(S,S')$-CRM mechanism is strategyproof when $S \subseteq S'$ or $S \cap S' = \emptyset$, and provided an involved, geometric proof. However, we have identified a mistake in their proof. In fact, we have found two counterexamples, one with $S \subseteq S'$ and one with $S \cap S' = \emptyset$, for which the corresponding $(S,S')$-CRM mechanisms violate strategyproofness, thus disproving their claim. These counterexamples are presented in Figure~\ref{fig:cnt}, 

\begin{figure}
	\centering
	\begin{subfigure}{0.5\linewidth}
		\centering
		\begin{tikzpicture}[scale=0.5]
		\begin{axis}[
		legend pos=outer north east
		]
		\addplot [only marks] table {
			1 0
			3 1
			5 1.9
		};
		\addlegendentry{Point in $S$}
		\addplot [only marks, mark=o] table {
			0 1
			2 2 
			4 3 
		};
		\addlegendentry{Point in $S'$}
		
		\addplot [only marks, mark = x] table{
			4 1.8
		};
		\addlegendentry{Deviation}

		\addplot [domain=-0.5:6, samples=2] {1};
		\addlegendentry{CRM before deviation}
		
		\addplot [domain=-0.5:6, samples=2, dashed] {0.1*x + 1.4};
		\addlegendentry{CRM after deviation}
		
		\end{axis}
		\end{tikzpicture}
		\caption{$S \cap S' = \emptyset$} \label{fig:cnt1}
	\end{subfigure}%
	\begin{subfigure}{0.5\linewidth}
		\centering
		\begin{tikzpicture}[scale=0.5]
		\begin{axis}[
		legend pos=outer north east
		]
		\addplot [only marks] table {
			3 12
			9 9.5
			11 9
			13 4.5
			14 11
		};
		\addlegendentry{Point in $S$}
		\addplot [only marks, mark=o] table {
			3 12
			4 8
			4.3 12
			7 6.5
			8 7.5
			9 9.5
			11 9
			12 11
			13 4.5
			14 11
		};
		\addlegendentry{Point in $S'\setminus S$}
		\addplot [only marks, mark = x] table{
			12 2
		};
		\addlegendentry{Deviation}

		\addplot [domain=0:15, samples=2] {0.5*x + 3.5};
		\addlegendentry{CRM before deviation}
		
		\addplot [domain=0:15, samples=2,dashed] {0.5833*x + 2.833};
		\addlegendentry{CRM after deviation}
		
		\end{axis}
		\end{tikzpicture}
		\caption{$S \subseteq S'$}\label{fig:cnt2}
	\end{subfigure}
	\caption{Counterexamples showing violation of strategyproofness of $(S,S')$-CRM mechanisms. Figure~\ref{fig:cnt1} shows a case with $S \cap S' = \emptyset$, while Figure~\ref{fig:cnt2} shows a case with $S \subseteq S'$.}
	\label{fig:cnt}
\end{figure}
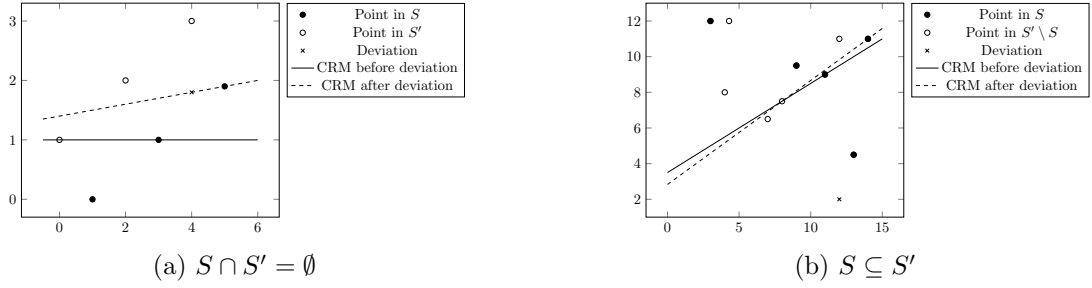

\begin{example}[Example with $S \cap S' = \emptyset$.]
	This example is shown in Figure~\ref{fig:cnt1}. Points in filled dots are in $S$, while points in empty dots are in $S'$. The coordinates of these points are as follows.
	\begin{align*}
	S = \left\{(1,0), (3,1), (5,1.9)\right\}, S' = \left\{(0,1),(2,2), (4,3)\right\}. 
	\end{align*}
	
	Notice that $S \cap S' = \emptyset$. Also, $|S|$ and $|S'|$ are odd, alleviating the need to choose between left and right medians in the CRM definition. 
	
	When the agents truthfully report, one can check that CRM returns the line connecting points $(3,1)$ from $S$ and $(0,1)$ from $S'$. This line is given by the equation $y=1$. 
	
	Suppose that the agent $i$ controlling the point at $x=4$ misreports $\ty_i= 1.8$ instead of $y_i = 3$. The new point is depicted with a cross. One can check that this causes the CRM mechanism to switch to the dashed line ($y = 0.1 \cdot x + 1.4$), which makes agent $i$ strictly better off, and violates strategyproofness.
\end{example} 

\begin{example}[Example with $S \subseteq S'$.]
	This example is shown in Figure~\ref{fig:cnt2}. Points in $S$ (thus also in $S'$) are depicted with filled dots, while points in $S'\setminus S$ are depicted with empty dots. The coordinates of these points are as follows. 
	\begin{align*}
	S = \left\{(3, 12), (9, 9.5), (11, 9), (13, 4.5), (14,11) \right\}, S' = S \cup \left\{(4,8), (4.3, 12), (7, 6.5), (8, 7.5), (12, 11) \right\}. 
	\end{align*}
	
	Notice that $S \subseteq S'$. Further, $|S|$ is odd, and $|S'|$ is even (thus, for each $i \in S$, $|S' \setminus \set{i}|$ is odd), once again eliminating the need to choose between the left and the right medians in the CRM definition. 
	
	When all points are reported truthfully, one can check that the CRM mechanism chooses the solid line ($3y = 2x+8$). Suppose now that agent $i$ with point $(12,11)$ reports $\ty_i = 0$, instead of $y_i = 11$. Then, the CRM mechanism chooses the dashed line, which makes agent $i$ strictly better off, again violating strategyproofness.
\end{example}

Nevertheless, we have been able to identify a subset of the CRM family, for which we can establish strategyproofness (in fact, group strategyproofness). In particular, we replace $S \subseteq S'$ with the more restrictive condition $S = S'$, and for $S \cap S' = \emptyset$, we either add $|S|=1$ or $|S'|=1$, or replace it with a stricter condition that we define below. 

\begin{definition}[Separable Sets of Points in a Plane]\label{def:sep}
	Let $S, S'$ be two sets of points in $\bbR^2$. We say that $S$ and $S'$ are \emph{separable} if $\max_{i \in S} x_i < \min_{j \in S'} x_j$ or $\max_{j \in S'} x_j < \min_{i \in S} x_i$. In other words, it should be possible to separate them by a vertical line. 
\end{definition}

Note that separability of $S$ and $S'$ implies $S \cap S' = \emptyset$. We now present a corrected version of the result of \citet{PP04}, and claim the stronger guarantee of group strategyproofness. We do not present a proof as we later introduce a much broader family of mechanisms, and prove their group strategyproofness directly. 

\begin{theorem}\label{thm:crm-corrected}
	Given $S,S' \subseteq N$, the $(S,S')$-CRM mechanism is group strategyproof if one of the following conditions holds.
	\begin{enumerate}
		\item $S = S'$.
		\item $S$ and $S'$ are separable.
		\item $S\cap S' = \emptyset$ and $\min(|S|,|S'|) = 1$.
	\end{enumerate}
\end{theorem}

The third condition partially resembles dictatorship as the agent in the singleton set is guaranteed to have zero residual (i.e., be on the regression line). 

\subsubsection{Generalized Resistant Line Mechanisms on a Plane}
\label{sec:gen-rl-plane}
In this section, our goal is to introduce a novel family of group strategyproof mechanisms that include, as special cases, the mechanisms covered in the three cases of Theorem~\ref{thm:crm-corrected}. Our starting point is the family of {\em resistant line} (RL) mechanisms from the statistics literature~\cite{JV85}, which \citet{PP04} showed to be equivalent to the case of separable $S$ and $S'$. 

The standard formulation of the RL mechanism involves three sets $L,M,R \subseteq N$ such that $\max_{i \in L} x_i < \min_{i \in M} x_i$ and $\max_{i \in M} x_i < \min_{i \in R} x_i$, and returns a line $\vbeta = (\beta_1,\beta_0)$ given by
$$
\textstyle\med_{i \in L} y_i - \beta_1 \cdot x_i - \beta_0 = \med_{i \in R} y_i - \beta_1 \cdot x_i - \beta_0 = 0.
$$
That is, the line makes the median residuals in $L$ and $R$ zero. It is known that this equation yields a unique solution~\cite{JV85}. \citet{PP04} showed that this is identical to the $(L,R)$-CRM mechanism. Indeed, separability of $L$ and $R$ makes clockwise angles from points in $L$ to points in $R$ monotonic in (and thus replaceable by) slopes, yielding the following formulation for the $(L,R)$-CRM mechanism.
\begin{align*}
\beta_1 &= \textstyle\med_{i \in L} \med_{j \in R} \frac{y_j-y_i}{x_j-x_i},\\
\beta_0 &= \textstyle\med_{i \in L} y_i - \beta_1 x_i = \med_{j \in R} y_j - \beta_1 x_j.
\end{align*}
The alternative definition of $\beta_0 = \med_{j \in R} (y_j - \beta_1 \cdot x_j)$ follows from the fact that if the line passes through $i^* \in L$, it is directed towards the point in $R$ which is at the median angle or slope, and thus bisects $R$ in addition to bisecting $L$.

Along with Theorem~\ref{thm:crm-corrected}, this observation establishes group strategyproofness of all resistant line mechanisms. Two popular mechanisms from this family are the Brown-Mood mechanism~\cite{BM51}, in which $L$ and $R$ each contain half of the points while $M$ is empty, and the Tukey mechanism~\cite{Tuk71}, in which $L$, $M$, and $R$ each contain a third of the points. 

Our next step is to extend this family. A natural idea is that instead of making the {\em median} residuals from $S$ and $S'$ zero, we make the $k^\th$ smallest residual in $S$ and the $(k')^\th$ smallest residual in $S'$ zero, for fixed $k \in [|S|]$ and $k' \in [|S'|]$.

\begin{definition}[Generalized Resistant Line (GRL) Mechanisms]
	Given separable sets $S,S' \subseteq N$, $k \in [|S|]$, and $k' \in [|S'|]$, the $(S,S',k,k')$-generalized resistant line (GRL) mechanism returns the line $\vbeta=(\beta_1,\beta_0)$ given by
	\begin{equation}\label{eq:gen-rl-plane}
	\textstyle\min_{i \in S}^k y_i -\beta_1 x_i - \beta_0 = \textstyle\min_{j \in S'}^{k'} y_j - \beta_1 x_j - \beta_0 = 0.
	\end{equation}
\end{definition} 

We show that these mechanisms are well defined (i.e., there is a unique solution to Equation~\eqref{eq:gen-rl-plane}), and they are group strategyproof. Once again, we omit the proof because we later introduce an even broader family of mechanisms, for which we prove these results directly. 

\begin{theorem}\label{thm:gen-rl-plane}
	For separable sets $S,S' \subseteq N$, $k \in [|S|]$ and $k' \in [|S'|]$, the $(S,S',k,k')$-generalized resistant line mechanism is well defined and group strategyproof.
\end{theorem}

While it is clear that generalized resistant line mechanisms cover the second case of Theorem~\ref{thm:crm-corrected} (i.e., separable $S$ and $S'$), we surprisingly find that they also cover the first case ($S=S'$) and the third case ($S\cap S' = \emptyset$ and $\min(|S|,|S'|) = 1$). That is, Theorem~\ref{thm:gen-rl-plane} strictly generalizes Theorem~\ref{thm:crm-corrected}. The proof of the next result is in Appendix~\ref{appendix:missing-proofs}.

\begin{lemma}\label{lem:crm-part-rl}
The $(S,S')$-CRM mechanism is a generalized resistant line mechanism when 
\begin{enumerate*}
\item $S = S'$,
\item $S$ and $S'$ are separable, or 
\item $S \cap S'=\emptyset$ and $\min(|S|,|S'|)=1$.
\end{enumerate*}
\end{lemma}

\subsubsection{Generalized Resistant Hyperplane Mechanisms in High Dimensions}
\label{sec:gen-rl}

Surprisingly, the statistics literature does not offer an extension of resistant line mechanisms to higher dimensions. In our efforts to do so, we quickly realized that this is a non-trivial task. In two dimensions, a generalized resistant line mechanism takes two subsets of data points separable by a vertical line, and returns the regression line which makes prescribed percentiles of residuals in each set zero. In $d+1$ dimensions (recall that $\vx(i) \in \bbR^d$ and $y_i \in \bbR$), it seems natural to take $d+1$ ``separable'' subsets of data points, and return the regression hyperplane which makes prescribed percentiles of residuals in each set zero. However, the separability condition must now ensure existence of a unique hyperplane with this property, even if we ignore our game-theoretic desiderata. 

In resolving this issue, we make a connection to the literature on the {\em Ham Sandwich Theorem} and its generalizations. Hereinafter, given a hyperplane $H$, we denote by $H^+$ and $H^-$ its positive and negative closed half-spaces, respectively. A basic version of the ham sandwich theorem due to~\citet{ST42} states that given $k$ continuous measures $\mu_1,\ldots,\mu_k$ on $\bbR^k$, there exists a hyperplane $H$ such that $\mu_i(H^+) = 1/2$ for each $i \in [k]$. A discrete version of this result due to~\citet{EH11} states that given $k$ finite sets $S_1,\ldots,S_k \subseteq \bbR^k$, there exists a hyperplane $H$ such that for each $i \in [k]$, $H$ ``bisects'' $S_i$ and $H \cap S_i \neq \emptyset$. Here, we say that a hyperplane $H$ bisects a set of points $S$ if each {\em closed} half-space of $H$ contains at least $\lceil |S|/2 \rceil$ points.

For linear regression, this implies that given $S_1,\ldots,S_{d+1} \subseteq \calD$, there exists a ``resistant hyperplane'' which makes the median residual from $S_t$ zero, for each $t \in [d+1]$. While this seems like a natural generalization of resistant line mechanisms, it is easy to check that such a hyperplane is not always unique, even in two dimensions. Further, if the median is replaced by other percentiles, the existence is no longer guaranteed.\footnote{Recall that even in two dimensions, we needed an additional condition on the sets $S$ and $S'$: separability by a vertical line.} 

\citet{SZ10} provide a generalization that {\em almost} perfectly fits our needs. They show that under certain conditions on $S_1,\ldots,S_{d+1}$, there exists a unique hyperplane $H$ which contains a given number of points from each set in its negative closed half-space. This discrete result builds upon previous continuous variants~\cite{BHJ08,Bre10}. We first define a condition they require, which also plays a key role in our result.

\begin{definition}[Well Separable Sets~\cite{KN73}]
	Given $t  \in [k+1]$, finite sets $S_1, \dots, S_t$ of points in $\bbR^k$ are called {\em well separable} if for all disjoint $I,J \subseteq [t]$, there exists a hyperplane $H$ such that $S_i \subset H^+ \setminus H$ for each $i \in I$ and $S_j \subset H^- \setminus H$ for each $j \in J$, i.e., $H$ separates $\cup_{i \in I} S_i$ from $\cup_{j \in J} S_j$ by putting them in different {\em open} half-spaces.
\end{definition}

Well separable sets are sometimes called {\em affinely independent} sets~\cite{Bre10}. Well separability is equivalent to various other conditions~\cite{Bre10,SZ10}. In what follows, $\conv(\cdot)$ denotes the convex hull.
\begin{proposition}\label{prop:well-separable}
	For $t \in [k+1]$, finite sets $S_1,\ldots,S_t \subset \bbR^k$ are well separable if and only if:
	\begin{enumerate}
		\item For all choices of $(x_i \in \conv(S_i))_{i \in [t]}$, the affine hull of $x_1, \ldots, x_t$ is a $(t-1)$-dimensional flat.
		\item No $(t-2)$-dimensional flat has a nonempty intersection with $\conv(S_i)$ for each $i \in [t]$.
		\item $\conv(S_1),\ldots,\conv(S_t)$ are well separable.
	\end{enumerate}
\end{proposition}

\citet{SZ10} impose an additional condition, which we eliminate in our work.

\begin{definition}[Weak General Position]
	Finite sets $S_1, \dots, S_k \subset \bbR^k$ are said to have {\em weak general position} if for every choice of $(x_i \in S_i)_{i \in [k]}$, the affine hull of $x_1,\ldots,x_k$ is a $(k-1)$-dimensional flat which contains no other point of $\cup_{i \in [k]} S_i$.
\end{definition}

\begin{theorem}[\cite{SZ10}]
	If finite sets $S_1,\ldots,S_k \subset \bbR^k$ are well separable and have weak general position, then given any choice of $k_i \in [|S_i|]$ for $i \in [k]$, there exists a unique hyperplane $H$ such that for each $i \in [k]$, $H \cap S_i \neq \emptyset$ and $|H^- \cap S_i| = k_i$. 
\label{thm:ham-sandwich}
\end{theorem}

This result gives us {\em almost} what we want for linear regression in $\bbR^{d+1}$. Given a family of sets $S_1,\ldots,S_{d+1} \subseteq \calD$ that are well separable and have weak general position, and $k_t \in [|S_t|]$ for $t \in [d+1]$, it ensures the existence of a unique hyperplane which makes the $k_t^\th$ smallest residual in each set $S_t$ zero. However, it falls short of our requirements in two key aspects.
\begin{itemize}
	\item Theorem~\ref{thm:ham-sandwich} allows the assignment of points in $\calD$ to sets $S_1,\ldots,S_{d+1}$ to depend on the private information $\vy$. For strategyproofness, we need this assignment to be based solely on the public information $\vx$. Recall that in two dimensions, we required sets $S$ and $S'$ to be separable by a {\em vertical} line. We choose the $d+1$ sets so that they are well separable in the $d$-dimensional public information space,\footnote{While Theorem~\ref{thm:ham-sandwich} uses $d+1$ well separable sets in $\bbR^{d+1}$, even $\bbR^d$ allows up to $d+1$ well separable sets.} and establish group strategyproofness using a technical lemma, which may be of independent interest.  
	\item While we only want to make the $k_t^\th$ smallest residual in each $S_t$ zero, \citet{SZ10} aim for something stronger: they want the number of points from each $S_t$ in the negative closed halfspace to be exactly $k_t$. This necessitates their weak general position assumption, which we relax.
\end{itemize}

We are now ready to present our results. They closely mirror, but do not make use of, the results of~\citet{SZ10}. We revert to using notation of our linear regression setting. Recall that a hyperplane $\vbeta = (\vbetaone,\beta_0)$ passes through $(\vx(i),\vbeta^T\ \bvx(i))$ for each $i \in N$, where $\bvx(i)=(\vx(i),1)$. %

\begin{definition}
	Given a family $\calS = (S_1,\ldots,S_k)$ of nonempty, pairwise disjoint subsets of $N$, and a set of points $P = (p_i)_{i \in N}$, define the partition function $\calP(P,\calS) = (P_t)_{t \in [k]}$, where $P_t = (p_i)_{i \in S_t}$ for each $t \in [k]$. That is, $\calP(P,\calS)$ partitions the set of points $P$ based on index sets from $\calS$.
\end{definition}

\begin{definition}[Publicly Separable Sets of Agents]
	We say that a family $\calS = (S_1,\ldots,S_{d+1})$ of nonempty, pairwise disjoint subsets of $N$ is {\em publicly separable} if $\calP(\vx,\calS)$ is well separable.
\end{definition}

\begin{definition}[Generalized Resistant Hyperplane (GRH) Mechanisms]
	Given a family $\calS = (S_1,\ldots,S_{d+1})$ of publicly separable sets of agents, and $\vk = (k_1,\ldots,k_{d+1})$ with $k_t \in [|S_t|]$ for $t \in [d+1]$, the $(\calS,\vk)$-generalized resistant hyperplane (GRH) mechanism returns a hyperplane $\vbeta$ such that $\min_{i \in S_t}^{k_t} (r_i \triangleq y_i - \vbeta^T\ \bvx(i)) = 0$ for each $t \in [d+1]$. That is, it makes the $k_t^\th$ smallest residual from every set $S_t \in \calS$ zero.
\end{definition}	

We first need to establish that the GRH mechanisms are well defined, i.e., the hyperplane they seek is guaranteed to exist and be unique. To that end, we prove a useful technical lemma, which may be of independent interest. 

\begin{lemma}[Hyperplane Comparison Lemma]\label{lem:hyperplane-comparison}
	Given a family $\calS = (S_1,\ldots,S_{d+1})$ of publicly separable sets of agents, and two distinct hyperplanes $\vbeta^1$ and $\vbeta^2$ in $\bbR^{d+1}$, there exists a set $S_t \in \calS$ such that either $(\vbeta^1)^T\ \bvx(i) < (\vbeta^2)^T\ \bvx(i)$ for all $i \in S_t$, or $(\vbeta^1)^T\ \bvx(i) > (\vbeta^2)^T\ \bvx(i)$ for all $i \in S_t$. 
\end{lemma}
\begin{proof}
	Consider the intersection of the two hyperplanes in $\bbR^{d+1}$, and let $W$ be its projection on $\bbR^d$ (the public information space). Note that $W$ is a $(d-1)$-dimensional hyperplane in $\bbR^d$. Given an {\em open} half-space of $W$ (say $W^+$), let $Z$ be the set of points $\bbR^{d+1}$ whose projection on $\bbR^d$ lies in $W^+$. Then, either $(\vbeta^1)^T\ \bar{\vec{p}} > (\vbeta^2)^T\ \bar{\vec{p}}$ for all $\vec{p} \in Z$, or $(\vbeta^1)^T\ \bar{\vec{p}} < (\vbeta^2)^T\ \bar{\vec{p}}$ for all $\vec{p} \in Z$, where $\bar{\vec{p}} = (\vec{p},1)$.

	Let $\calP(\vx,\calS) = (X_1,\ldots,X_{d+1})$. Because $\calS$ is publicly separable, $X_1,\ldots,X_{d+1}$ are well separable. By Proposition~\ref{prop:well-separable}, no $(d-1)$-dimensional flat has a nonempty intersection with $\conv(X_t)$ for each $t \in [d+1]$. Because $W$ is a $(d-1)$-dimensional flat, there exists $t \in [d+1]$ such that $W$ does not intersect $\conv(X_t)$, i.e., $X_t$ lies entirely in an {\em open} half-space of $W$. Using the previous argument, either $(\vbeta^1)^T\ \bvx(i) < (\vbeta^2)^T\ \bvx(i)$ for all $i \in S_t$, or $(\vbeta^1)^T\ \bvx(i) > (\vbeta^2)^T\ \bvx(i)$ for all $i \in S_t$. 
\end{proof}

\begin{proposition}\label{prop:gen-rl-valid}
	Generalized resistant hyperplane mechanisms are well defined. That is, given a family $\calS = (S_1,\ldots,S_{d+1})$ of publicly separable sets of agents, and $\vk = (k_1,\ldots,k_{d+1})$ with $k_t \in [|S_t|]$ for $t \in [d+1]$, there exists a unique hyperplane $\vbeta$ for which $\min_{i \in S_t}^{k_t} y_i - \vbeta^T\ \bvx(i) = 0$ for each $t \in [d+1]$.
\end{proposition}
\begin{proof}
	First, we show that {\em if} such a hyperplane exists, it must be unique. Suppose for contradiction that there are two distinct hyperplanes $\vbeta^1$ and $\vbeta^2$ which make the $k_t^\th$ smallest residual from every $S_t \in \calS$ zero. By the hyperplane comparison lemma (Lemma~\ref{lem:hyperplane-comparison}), there exists $S_t \in \calS$ such that either $(\vbeta^1)^T\ \bvx(i) < (\vbeta^2)^T\ \bvx(i)$ for all $i \in S_t$, or $(\vbeta^1)^T\ \bvx(i) > (\vbeta^2)^T\ \bvx(i)$ for all $i \in S_t$. Without loss of generality, suppose it is the former. Then, at least $k_t$ points in $S_t$ which have a non-positive residual under $\vbeta^2$ have a negative residual under $\vbeta^1$, contradicting the fact that $\vbeta^1$ makes the $k_t^\th$ smallest residual from $S_t$ zero.
	
	For proving existence, we use a counting technique. Create two bipartite graphs $G = (V \cup W, E)$ and $G' = (V' \cup W, E')$. Let $V$ (resp. $V'$) contain a vertex $v_{\vk}$ (resp. $v'_{\vk}$) corresponding to each $\vk = (k_1,\ldots,k_{d+1})$ such that $k_t \in [|S_t|]$ for each $t \in [d+1]$. Thus, $|V| = |V'| = \prod_{t=1}^{d+1} |S_t|$. Let $W$ contain a vertex $w_{\vbeta}$ corresponding to every {\em traversal} hyperplane $\vbeta$, i.e., every hyperplane that passes through at least one point from each set $S_t \in \calS$. 
	
	In graph $G$, we draw an edge between $v_{\vk}$ and $w_{\vbeta}$ if $\vbeta$ makes the $k_t^\th$ smallest residual zero in each $S_t \in \calS$. For constructing graph $G'$, we fix an arbitrary ordering of points in each set, so that we can write $S_t = \set{i^t_1,\ldots,i^t_{|S_t|}}$. Then, we draw an edge in $G'$ between $v'_{\vk}$ and $w_{\vbeta}$ if $\vbeta$ passes through point $i^t_{k_t}$ for each $t \in [d+1]$. %
	 	
	Our goal is to show that each vertex $v_{\vk} \in V$ has exactly one incident edge in graph $G$. We prove this through a sequence of claims. First, we argue that each vertex $v'_{\vk} \in V'$ has exactly one incident edge in graph $G'$. The fact that it has {\em at least} one incident edge follows from the fact that any set of $d+1$ points in $\bbR^{d+1}$ (in particular, $T = \set{i^t_{k_t}}_{t \in [d+1]}$) lie on a hyperplane. If $v'_{\vk}$ has two or more incident edges, then there exist two distinct hyperplanes $\vbeta^1$ and $\vbeta^2$ which pass through all points in $T$. Then, their intersection $\vbeta^*$, which is a $(d-1)$-dimensional flat in $\bbR^{d+1}$, must also pass through all points in $T$. Let $\calP(\vx,\calS) = (X_1,\ldots,X_{d+1})$. Then, the projection of $\vbeta^*$ on the public information space $\bbR^d$ is a $(d-1)$-dimensional hyperplane in $\bbR^d$ which intersects each $X_t$ (and thus each $\conv(X_t)$). However, $\calS$ is a publicly separable family, i.e., $X_1,\ldots,X_{d+1}$ are well separable in $\bbR^d$. This violates the first condition of Proposition~\ref{prop:well-separable}. 
	
	Since each vertex in $V'$ has exactly one incident edge, we have $|E'| = |V'| = \prod_{t=1}^{d+1} |S_t|$. We next argue that $|E| = |E'|$. Take a vertex $w_{\vbeta} \in W$. Note that if hyperplane $\vbeta$ passes through $a_t$ points from each $S_t \in \calS$, then it has degree $\prod_{t=1}^{d+1} a_t$ in both $G$ and $G'$. Since each vertex in $W$ has the same degree in both graphs, we have $|E| = |E'| = |V'| = |V|$. 
	
	Finally, we already established that if there is a hyperplane which makes the $k_t^\th$ smallest residual in each $S_t$ zero, then it must be unique. Thus, each vertex in $V$ has {\em at most} one incident edge in $G$. Together with $|E| = |V|$, this implies that each vertex in $V$ has {\em exactly} one incident edge in $G$.
\end{proof}

We are now ready to present our main contribution. 

\begin{theorem}\label{thm:gen-rl-gsp}
	Every generalized resistant hyperplane mechanism is group strategyproof.
\end{theorem}
\begin{proof}
	Consider an $(\calS,\vk)$-generalized resistant hyperplane mechanism. Consider a set of data points $\calD = (\vx(i),y_i)_{i \in N}$. Suppose a coalition $S \subseteq N$ of agents changes their report to $(\ty_i)_{i \in S}$, and changes the resulting hyperplane from $\vbeta$ to $\tilde{\vbeta}$. Set $\ty_i = y_i$ for $i \in N\setminus S$, and let $\tilde{\calD} = (\vx(i),\ty_i)_{i \in N}$. 
		
	By the hyperplane comparison lemma (Lemma~\ref{lem:hyperplane-comparison}), there exists $S_t \in \calS$ such that either $\vbeta^T\ \bvx(i) < \tilde{\vbeta}^T\ \bvx(i)$ for all $i \in S_t$, or $\vbeta^T\ \bvx(i) > \tilde{\vbeta}^T\ \bvx(i)$ for all $i \in S_t$. 
	
	Without loss of generality, suppose it is the former. The $k_t^\th$ smallest residual from $S_t$ is zero under $\vbeta$ in $\calD$, and under $\tilde{\vbeta}$ in $\tilde{\calD}$. If $S \cap S_t = \emptyset$, or if every manipulator in $S \cap S_t$ has a positive residual under $\vbeta$ in $\calD$, then at least $k_t$ non-manipulators in $N\setminus S$ have a non-positive residual under $\vbeta$ in $\calD$, and thus a strictly negative residual under $\tilde{\vbeta}$ in $\tilde{\calD}$, which contradicts the fact that $\tilde{\vbeta}$ makes the $k_t^\th$ smallest residual in $S_t$ zero in $\tilde{\calD}$. 
	
	In other words, there must exist a manipulator $i \in S \cap S_t$ who has a non-positive residual under $\vbeta$ in $\calD$. Thus, $\tilde{\vbeta}^T\ \bvx(i) > \vbeta^T\ \bvx(i) \ge y_i$, implying that the manipulator is strictly worse off after the manipulation. Hence, the mechanism is group strategyproof.
\end{proof}

For two dimensions ($d=1$), we already argued that our sub-family of group strategyproof CRM mechanisms given by Theorem~\ref{thm:crm-corrected} is part of the larger family of GRL mechanisms (Lemma~\ref{lem:crm-part-rl}). It is easy to see that GRL mechanisms are precisely GRH mechanisms in two dimensions. Indeed, GRH mechanisms would require two subsets of agents $S_1,S_2$ that are publicly separable, i.e., well separable on the $x$-axis. Note that this coincides with the separability definition used by GRL mechanisms (Definition~\ref{def:sep}). Hence, the $(S,S',k,k')$-GRL mechanism is precisely the $(\calS,\vk)$-GRH mechanism with $\calS = (S,S')$ and $\vk = (k,k')$. In three or more dimensions, we do not know if, given $\vx$, one can always construct a family $\calS$ of publicly separable sets of agents such that each set $S_t \in \calS$ contains at least a constant fraction of the agents.

\subsection{Strategyproofness vs Group Strategyproofness}
\label{sec:gsp}

In the single dimensional setting ($d=0$), \citet{Moul80} proved that all strategyproof mechanisms are also group strategyproof. This alternatively follows from a result by~\citet{BBM10}, who gave a sufficient condition on the underlying domain for the sets of strategyproof and group strategyproof mechanisms to coincide. 

Interestingly, all known strategyproof mechanisms for the multidimensional linear regression setting (including generalized $\erm$ and generalized resistant hyperplane mechanisms) are group strategyproof as well. However, it is easy to check that the linear regression setting does not satisfy the sufficient condition of~\citet{BBM10}. Is it still true that all strategyproof mechanisms for linear regression are also group strategyproof? We answer this question {\em negatively}. 

\begin{example}\label{ex:sp-not-gsp} 
	Consider the simple linear regression setting ($d=1$) with $n=2$ agents. Fix the public information $\vx = (x_1,x_2) \in \bbR^2$, and consider the mechanism $M$ that, on input $\vy = (y_1,y_2)$, returns the line passing through points $(x_1,y_2)$ and  $(x_2,y_1)$. Under this mechanism, the outcome for each agent is independent of the agent's report: indeed, the outcome for agent $1$ (resp. agent $2$) is $\hy_1 = y_2$ (resp. $\hy_2 = y_1$). Hence, the mechanism is clearly strategyproof. However, group strategyproofness is violated because when $y_1 \neq y_2$, the two agents can collude, and report $\tvy = (y_2,y_1)$. This makes the resulting line pass through both agents, making both strictly better off.
\end{example}

The requirement that the outcome for each agent be independent of the agent's report, called {\em impartiality} in mechanism design, is stricter than (i.e., logically implies) strategyproofness, %
and has been studied for aggregating opinions or dividing rewards~\cite{CMT08,HM13,TO14,FK15,KLMP15}.

\begin{definition}[Impartial Mechanisms]
	A mechanism $M$ is called \emph{impartial} if the outcome for each agent is independent of the agent's report. Formally, for every agent $i \in N$, reports $\vy$, and alternative report $y'_i$ by agent $i$, we require that $\hy_i(M(\vy)) = \hy_i(M(y'_i,\vy_{-i}))$. 
\end{definition}

In linear regression, when the number of agents is $n=d+1$, we can easily characterize all impartial mechanisms because we can set $\hy_i$ to be an arbitrary function of $\vy_{-i}$, and return a hyperplane passing through the resulting $d+1$ points $(\vx(i),\hy_i)_{i \in N}$. 

\begin{proposition}
	For $n=d+1$, mechanism $M$ is impartial if and only if there exist functions $f_1,\ldots,f_n : \bbR^{n-1} \to \bbR$ such that given $\vy$, $M$ returns a hyperplane passing through $(\vx(i),f_i(\vy_{-i}))_{i \in N}$.
\label{prop:impartial-few-agents}
\end{proposition}

Note that functions $f_i$ can even be discontinuous, which can make the regression hyperplane discontinuous in the input $\vy$. However, we later show (Theorem~\ref{thm:our-char}) that under any strategyproof mechanism, the outcome $\hy_i$ for agent $i$ must be a continuous function of $y_i$ (it is a constant function of $y_i$ in case of impartial mechanisms). 

With $n > d+1$ points, the question of whether impartial mechanisms even exist is non-trivial. While we still need to set each $\hy_i$ as a function of $\vy_{-i}$, it cannot be done arbitrarily as the resulting points $(\vx(i),\hy_i)_{i \in N}$ may no longer lie on a hyperplane. In other words, setting $\hy_i$ as a function of $\vy_{-i}$ for $d+1$ agents already determines the hyperplane, and thus $\hy_j$ for all remaining agents $j$. The mechanism must ensure that these $\hy_j$ are also independent of $y_j$. At first glance, this may seem impossible, except in the trivial case where a constant hyperplane is returned regardless of $\vy$. 

Nonetheless, we show that there exists a wide family of non-trivial impartial mechanisms for linear regression. Our family provides a full characterization of impartial mechanisms for $d=1$ (i.e., for simple linear regression). In the result below, we use the notation $\lbr \vec{a},\vec{b} \rbr$ instead of $\vec{a}^T \vec{b}$ for the sake of simplicity. Its proof is in Appendix~\ref{appendix:missing-proofs}.

\begin{theorem}\label{thm:impartial}
	Given $\vx$, mechanism $M^{\vx}$ for linear regression is impartial if there exist functions $\set{g_i^{\vx} : \bbR \to \bbR^d}_{i \in N}$ and constant $c^{\vx} \in \bbR$ such that for all $\vy$, we have $M^{\vx}(\vy) = \vbeta = (\vbetaone,\beta_0)$, where
	\begin{equation}
	\vbetaone = \textstyle\sum_{i \in N} g_i^{\vx}(y_i), \quad \beta_0 = c^{\vx}-\textstyle\sum_{i \in N} \lbr g_i^{\vx}(y_i),\vx(i) \rbr.
	\label{eq:impartial-hyperplane}
	\end{equation}
	For $d=1$ and an admissible set of points, this characterizes all impartial mechanisms.
\end{theorem}

Impartial mechanisms are not compelling from a statistical viewpoint. For instance, in the standard two-dimensional stochastic model where the data points are assumed to be generated by taking points on an underlying line and introducing i.i.d. errors in the dependent variables, it is easy to show that no impartial mechanism can produce an unbiased estimator of the underlying line. Nonetheless, impartial mechanisms help us establish the existence of a rather wide family of strategyproof mechanisms that are {\em not} group strategyproof. In fact, the next result shows that almost all impartial mechanisms violate group strategyproofness; its proof is in Appendix~\ref{appendix:missing-proofs}.

\begin{proposition}\label{prop:impartial-not-gsp}
	For simple linear regression ($d=1$) with an admissible set of points, an impartial mechanism is group strategyproof if and only if it is a constant function (i.e., it returns a fixed regression line regardless of its input).
\end{proposition}

\section{Characterizing Strategyproof Mechanisms}
\label{sec:characterization}

As mentioned in Section~\ref{subsec:l1-ERM}, \citet{Moul80} studied the one-dimensional setting ($d=0$), and analytically characterized all strategyproof mechanisms for $n$ agents. While we are unable to provide an analytical characterization for multidimensional linear regression, we provide two non-constructive characterizations, and discuss their implications. 

Interestingly, to characterize strategyproof mechanisms for linear regression with $n$ agents, we use the characterization of strategyproof mechanisms for the one-dimensional setting with a single agent. In this case, \citet{Moul80} shows that a mechanism is strategyproof if and only if there exist constants $\alpha^1,\alpha^2 \in \bbRbar$ such that when the agent reports $y$, the mechanism returns $\hy = \med(y,\alpha^1,\alpha^2)$. Constants $\alpha^1$ and $\alpha^2$ are called {\em phantoms}. First, we extend this result by providing an alternative characterization, which uses the following definition. The proof of the next result is in Appendix~\ref{appendix:missing-proofs}.

\begin{definition}[Locally Constant Function]
	For $A,B \subseteq \bbR$, function $f: A \to B$ is called locally constant at $x \in A$ if there exists $\eps > 0$ such that for all $x' \in [x-\eps,x+\eps]$, $f(x') = f(x)$. 
\end{definition}

\begin{lemma}\label{lem:1dim-1agent}
	Suppose mechanism $\pi : \bbR \to \bbR$ for the one-dimensional setting with a single agent elicits private value $y$ from the agent and returns $\pi(y)$. Then, $\pi$ being strategyproof is equivalent to each of the following conditions.
	\begin{enumerate}[label=(\alph*)]
		\item\label{thm1agent-parta} There exist constants $\alpha^1,\alpha^2 \in \bbRbar \triangleq \bbR \cup \set{-\infty,\infty}$ such that for all $y \in \bbR$, $\pi(y) = \med(y,\alpha^1,\alpha^2)$.
		\item\label{thm1agent-partb} $\pi$ is continuous, and for every $y \in \bbR$, either $\pi(y) = y$ or $\pi$ is locally constant at $y$. 
	\end{enumerate}
\end{lemma}

In the one-dimensional setting, \citet{Moul80} observed that a mechanism is strategyproof if and only if its outcome is strategyproof in the report of each individual agent when other agents' reports are fixed. That is, a mechanism $\pi:\bbR^n \to \bbR$ for $n$ agents is strategyproof if and only if
\begin{align}
&\forall i \in [n],\ \ \exists \alpha^1_i,\alpha^2_i \in \bbRbar \text{ independent of } y_i\ \text{ s.t. }\ \pi(y_1,\ldots,y_n) = \med(y_i,\alpha^1_i,\alpha^2_i).
\label{eq:moulin-char}
\end{align}
\citet{Moul80} solved Equation~\eqref{eq:moulin-char} to derive an elegant analytical expression for $\pi$ in terms of $\set{y_i}_{i \in [n]}$. Note that in this equation, the outcome $\hy = \pi(y_1,\ldots,y_n)$ is common to all agents. 

In contrast, in linear regression each agent $i$ has a potentially different outcome $\hy_i$. Like before, strategyproofness requires that each $\hy_i$ obey the conditions in Lemma~\ref{lem:1dim-1agent}, when seen as a function of $y_i$, when other agents' reports are fixed. However, the outcomes for different agents are now constrained so that $(\vx(i),\hy_i)_{i \in N}$ lie on a hyperplane. This added complexity prevented us from solving the equations to derive an analytical characterization, despite significant effort. The only exception was the special case of {\em impartial} mechanisms, where we further restrict $\hy_i$ to be independent of $y_i$ (Theorem~\ref{thm:impartial}). This corresponds to the case where $\alpha^1_i=\alpha^2_i$ for each agent $i$. Nonetheless, by simply applying Lemma~\ref{lem:1dim-1agent} for every agent $i$, we obtain the following non-constructive characterization of strategyproof mechanisms for linear regression.

\begin{theorem}\label{thm:our-char}
	Given public information $\vx$, mechanism $M^{\vx}$ for linear regression being strategyproof is equivalent to each of the following conditions.  
	\begin{enumerate}[label=(\alph*)]
	\item For every $\vy_{-i} \in \bbR^{n-1}$ and $i \in N$, there exist $\ell_i,h_i \in \bbRbar$ such that $\hy_i(M^{\vx}(\vy)) = \med(y_i,\ell_i,h_i)$ for all $y_i \in \bbR$;
	\item For every $\vy_{-i} \in \bbR^{n-1}$ and $i \in N$, function $f_i(\cdot) = \hy_i(M(\cdot,\vy_{-i}))$ is continuous, and for every $y_i \in \bbR$, either $f_i(y_i) = y_i$ or $f_i$ is locally constant at $y_i$.
	\end{enumerate}
\end{theorem}

The first condition provides an analytical form of $\hy_i$ in terms of $y_i$, and is perhaps the more useful characterization. For instance, we crucially use this characterization in the next section to give a lower bound on the efficiency of strategyproof mechanisms. Our earlier (more complex) proof of group strategyproofness of GRH mechanisms (Theorem~\ref{thm:gen-rl-gsp}) was also based on this condition, and identified the precise $\ell_i$ and $h_i$ for each agent $i$. 

Note that for fixed $\vy_{-i}$, we have $\hy_i = y_i$ when $y_i \in [\ell_i,h_i]$. For $y_i \le \ell_i$, $\hy_i = \ell_i$ is fixed, and for $y_i \ge h_i$, $\hy_i = h_i$ is fixed. We therefore say that agent $i$ is {\em influential} over the interval $(\ell_i,h_i)$, and call $\ell_i$ and $h_i$ the {\em lower} and {\em upper influence bounds}, respectively. Analysis of influence bounds has received attention in the statistics literature, where it is called {\em sensitivity analysis}. For instance, \citet{NW85} observed that under $\erm$, the regression hyperplane is unaffected when the dependent variable of a point is changed so that the point still lies on the same side of the hyperplane as before. From Theorem~\ref{thm:our-char}, we can see that for every strategyproof mechanism, doing so should at least keep the outcome for agent $i$ unchanged. \citet{NW85} also focused on computing the influence bounds. Theorem~\ref{thm:our-char} lends a simple algorithm to compute influence bounds (see Appendix~\ref{appendix:influence-bounds}).
Finally, note that while $\hy_i$ must be continuous in $y_i$, it need not be continuous in $\vy$ (see our discussion on Proposition~\ref{prop:impartial-few-agents}).

\section{Efficiency of Strategyproof Mechanisms}
\label{sec:errors}

Insofar, we studied families of strategyproof mechanisms for linear regression. %
In the absence of strategic considerations, a popular mechanism for linear regression is the OLS (ordinary least squares), which is the empirical risk minimizer for the squared loss. Under this loss function, which is also called the {\em residual sum of squares} ($\rss$), the loss when choosing hyperplane $\vbeta$ given data points $\calD$ is 
$
 \rss(\calD,\vbeta) = \sum_{i \in N} \left(y_i - \vbeta^T\ \bvx(i)\right)^2.
$
A classic justification for the OLS is due to the Gauss-Markov theorem, which states that when the errors (deviations of data points from an underlying hyperplane we wish to identify) are stochastic, zero in expectation, uncorrelated, and of equal variance, the OLS is the {\em best linear unbiased estimator}. %

However, in our strategic setting, the OLS is not strategyproof~\cite{DFP10}. This raises an important question: {\em Is there a strategyproof mechanism that is close to the OLS?} We assess this by the worst-case approximation ratio of a mechanism for the optimal squared loss.

\begin{definition}[Efficiency]
	Given $\vx$, we say that mechanism $M^{\vx}$ for linear regression is $c$-efficient if for every $\calD = (\vx(i),y_i)_{i \in N}$, we have
	$
	\rss(\calD,M^{\vx}(\vy)) \le c \cdot \inf_{\vbeta} \rss(\calD,\vbeta). 
	$
\end{definition}

We show that no strategyproof mechanism that is too close to the OLS can be strategyproof. The proof of the next result leverages our characterization of strategyproof mechanisms (Theorem~\ref{thm:our-char}).

\begin{theorem}\label{thm:efficiency-lower-bound}
	For $n \ge 4$, there exist $\vx$ for which no strategyproof mechanism is $(2-\epsilon)$-efficient for any $\epsilon > 0$.
\end{theorem}
\begin{proof}
	For simplicity of notation, we use $n+1$ agents instead of $n$ agents (and assume $n+1 \ge 4$, i.e., $n \ge 3$). We also consider simple linear regression ($d=1$); the proof easily extends to higher dimensions by simply setting all other coordinates to zero. Fix $n \ge 3$. Consider a setting with $n+1$ agents where $x_i = i$ for $i \in [n]$, and $x_{n+1} = X$, where $X$ is the solution of the following equation:
	\begin{equation}
	\frac{n^3-n}{2(1+3n+2n^2+6X^2-6Xn-6X)} = 1.
	\label{eqn:X}
	\end{equation}
	Interested readers may note that $X = \Theta(n^{1.5})$. Let $T$ denote the LHS in Equation~\eqref{eqn:X}. 
	
	Consider a strategyproof mechanism $M^{\vx}$. Suppose $M^{\vx}$ is $c$-efficient. We want to show that $c \ge 2$. We consider a family of inputs $\vy$, in which we fix $y_i = 0$ for $i \in [n]$, and vary $y_{n+1} = Y$. First, we note that the optimal $\rss$, as a function of $Y$, is given by
	$$
	f_0(Y) = Y^2 \cdot \frac{n^3-n}{2+5n+4n^2+n^3-12X-12nX+12X^2} = Y^2 \cdot \frac{T}{T+1} = \frac{Y^2}{2},
	$$
	where the first transition is obtained by minimizing $(Y-X \cdot \beta_1 -\beta_0)^2 + \sum_{i=1}^n (i \cdot \beta_1+\beta_0)^2$ over all $(\beta_1,\beta_0)$, the second transition follows through simple algebra, and the final transition follows from Equation~\eqref{eqn:X}. For verification of these claims through Mathematica, see Figure~\ref{fig:mathematica} in Appendix~\ref{appendix:missing-proofs}.
	
	Recall that we fixed $y_i$ for $i \in [n]$. Due to our characterization result (Theorem~\ref{thm:our-char}), there exist $\ell,h \in \bbRbar$ with $\ell \le h$ such that the line returned by the mechanism passes through $(X,\med(Y,\ell,h))$ for all $Y$. We take two cases. 
	
	{\em Case 1: $h > 0$.} Set $Y = h$. Then, the line returned by the mechanism  passes through $(X,h)$. In this case, we can show that the $\rss$ of the mechanism is at least 
	$$
	f_1 = h^2 \cdot \frac{n^3-n}{2(1+3n+2n^2+6X^2-6Xn-6X)} = h^2 \cdot T = h^2,
	$$
	where the first transition is obtained by minimizing $(Y-\beta_1\cdot X-\beta_0)^2 + \sum_{i=1}^n (\beta_1\cdot i+\beta_0)^2$ over all $(\beta_1,\beta_0)$ which satisfy $\beta_1 \cdot X + \beta_0 = Y$, and the rest follows from Equation~\eqref{eqn:X}. For verification of these claims through Mathematica, see Figure~\ref{fig:mathematica} in Appendix~\ref{appendix:missing-proofs}. 
This implies $c \ge f_1/f_0(h) = 2$.
	
	{\em Case 2: $h \le 0$.} Set $Y=1$. Then, the line returned by the mechanism passes through $(X,h)$. In this case, the $\rss$ of the mechanism is at least $f_2 = 1$ because agent $n+1$ contributes $(1-h)^2 \ge 1$ to the squared loss. Once again, we have $c \ge f_2/f_0(1) = 2$. 
	
	The proof is complete as we have $c\ge 2$ in each case.
\end{proof}

For $n=2$ agents (or $n = d+1$ agents in $d+1$ dimensions), there is an obvious $1$-efficient strategyproof mechanism which returns a hyperplane passing through all input points. Theorem~\ref{thm:our-char} leaves open the case of $n=3$ in two dimensions.

\section{Discussion}
\label{sec}

Our work leaves several open questions. Perhaps the most ambitious one is to find a constructive characterization of all strategyproof or group strategyproof mechanisms for linear regression, which may allow us to pinpoint the most efficient strategyproof mechanism; \citet{CPS16b} provide a similar analysis in the one-dimensional setting. It is easy to show that $\erm$ is $n$-efficient (see Proposition~\ref{prop:erm-efficient} in Appendix~\ref{appendix:missing-proofs}). Does there exist a more efficient strategyproof mechanism? It would also be interesting to analyze efficiency in a stochastic setting where the data points are drawn from an underlying distribution. 

The characterization result of \citet{Moul80} for strategyproof and anonymous mechanisms in the one-dimensional setting extends the median to generalized medians by adding fixed phantom values, and then taking the median. It is also shown that adding $n+1$ phantoms is sufficient to obtain full generality. We can extend all our proposed families of mechanisms by adding a certain number of ``phantom points'' in $\bbR^{d+1}$, and then applying the mechanisms to the union of data points and phantom points. The resulting mechanism retains the incentive guarantees.\footnote{We also considered adding phantom values directly in the equations where a median is used. However, most such attempts violated strategyproofness.} Given $n$ data points, how many phantoms are sufficient to obtain full generality? Do the phantoms play a role in obtaining the elusive constructive characterization?

Another interesting observation is that our generalized resistant hyperplane mechanisms are guaranteed pass through $d+1$ input points in $d+1$ dimensions. %
It is known that at least one minimizer of the $L1$ loss also has this property. It would be interesting to identify a generic family of conditions, which, when imposed in addition to the requirement of making $d+1$ residuals zero, yield group strategyproofness.

Finally, \citet{DFP10} study a regression setting in which a single agent may control multiple data points, show that $\erm$ is no longer strategyproof, and provide novel strategyproof mechanisms. It would be useful to see if our ideas can be used to design additional strategyproof mechanisms in this model. Another interesting variant is when only a small number of data points are held by strategic agents, but the mechanism does not know which ones. A similar setting was studied by~\citet{CSV17}, but for classification and with adversarial manipulations. On a high level, we view our work as a stepping stone to studying incentives in more realistic machine learning environments.

\bibliographystyle{acm}
\bibliography{abb,ultimate,refs}

\clearpage
\appendix
\section*{Appendix}
\section{Missing Results and Proofs}
\label{appendix:missing-proofs}

In this section, we present the results and proofs missing from the main body of the paper.

\subsection{Generalized $\erm$ is Group Strategyproof}

\begin{proof}[Proof of Theorem~\ref{thm:gen-l1-gsp}]

We will follow the structure of the proof presented by \citet{DFP10}.

\begin{proposition}\label{prop:dfp}
Let $\hat{S} = \{(\vx_i, \hy_i)\}_{i=1}^m$ and $\tilde{S} = \{(\vx_i,\ty_i)\}_{i=1}^m$ be two training sets on the same set of points and let $\hat{f} = \texttt{w-ERM-reg}(\mathcal{F},\ell,\hat{S})$ and $\tilde{f} = \texttt{w-ERM-reg}(\mathcal{F}, \ell,\tilde{S})$, where by $\texttt{w-ERM-reg}$ we denote the weighted $L_1$-ERM with convex regularizer (i.e., $\tf = \arg\max_f \sum_{i \in N} w_i^{\vx} |y_i - f(\vx_i)| + h(f)$ ) and by $\ell$ the $L_1$ loss function. If $\hat{f} \neq \tilde{f}$ then, there exists $i \in N$, such that $\hy_i \neq \ty_i$ and 
\begin{equation}\label{eq:loss}
\ell(\hat{f}(\vx_i),\hat{y}_i) < \ell(\tilde{f}(\vx_i),\hat{y}_i)
\end{equation} 
\end{proposition}
\begin{proof}
Let $U = \{i: \hy_i \neq \ty_i\}$ and assume that $\ell(\hat{f}(\vx_i),\hat{y}_i) \geq \ell(\tilde{f}(\vx_i),\hat{y}_i)$ for all $i \in U$. First, we will consider functions of the form $f_{\alpha}(\vx) = \alpha\tilde{f}(\vx) + (1-\alpha)\hat{f}(\vx)$ and prove that there exists $\alpha \in (0,1]$ such that: 
\begin{equation}\label{eq:risk}
\hat{R}\left(\hf, \tS\right) - \hat{R}\left(\hf,\hS\right) = \hat{R}\left(f_{\alpha},\tS\right) - \hat{R}\left(f_{\alpha},\hS\right)
\end{equation}

For all $i \in U$ from Equation~\eqref{eq:loss} we get that either of the four inequalities below holds:
\begin{align}
\tf(\vx_i) \leq \hy_i < \hf(\vx_i), \quad\tf(\vx_i) \geq \hy_i > \hf(\vx_i) \label{eq:loss-case1}\\
\hy_i \leq \tf(\vx_i) < \hf(\vx_i),   \quad\hy_i \geq \tf(\vx_i) \geq  \hf(\vx_i) \label{eq:loss-case2} 
\end{align}

Observe now, that similarly to \citet{DFP10} since $\ty_i \neq \tf(\vx_i)$ produces the least sum of the weighted loss and the convex regularizer, then assuming that $\ty_i = \tf(\vx_i)$ will cause greater risk reduction for $\tf$, and therefore, $\tf$ will still minimize the risk. If one of the two inequalities in Equation~\eqref{eq:loss-case1} holds: 
\begin{equation}\label{eq:alpha}
\alpha_i = \frac{\hy_i - \hf(\vx_i)}{\tf(\vx_i) - \hf(\vx_i)}
\end{equation}
where $\alpha_i \in (0,1]$ and $f_{\alpha_i}(\vx_i) = \hy_i$. By substituting, for every $\alpha \in (0, \alpha_i]$ it holds that: 
\begin{equation*}
\ty_i \leq \hy_i \leq f_{\alpha_i}(\vx_i) < \hf(\vx_i) \quad \text{or} \quad \ty_i \geq \hy_i \geq f_{\alpha_i}(\vx_i) > \hf(\vx_i)
\end{equation*}
Based on the above and if we set $c_i = |\hy_i - \ty_i|$, we have that for all $\alpha \in (0,\alpha_i]$:

\begin{equation}\label{eq:c_i}
\ell (\hf(\vx_i), \ty_i) - \ell (\hf(\vx_i), \hy_i) = c_i \quad \text{and} \quad \ell(f_{\alpha}(\vx_i), \ty_i) - \ell(f_{\alpha}(\vx_i), \hy_i) = c_i
\end{equation}

By using Equation~\eqref{eq:loss-case2} and setting $\alpha_i = 1$ and $c_i = -|\ty_i - \hy_i|$ one ends up again with Equation~\eqref{eq:c_i}. Equation~\eqref{eq:c_i} holds for every $i \in U$ if we set $\alpha = \min_{i \in U} \alpha_i$ and it trivially holds for all $i \neq U$ with $c_i = 0$. Multiplying with the appropriate weights $w_i^{\vx}$ the equalities for each $i$ and summing them all, one gets to Equation~\eqref{eq:risk}. Note that in this step, the regularizer $h(f)$ can be ignored, since it cancels out from each side of the equation. 

Since $\mathcal{F}$ is a convex set, $f_{\alpha} \in \mathcal{F}$. Since $\hf$ minimizes the empirical risk with respect to $\hS$ over $\mathcal{F}$ we have that $\hat{R}(\hf, \hS) \leq \hat{R}(f_{\alpha}, \hS)$ and combining with Equation~\eqref{eq:loss} we get that $\hat{R}(\hf, \tS) \leq \hat{R}(f_{\alpha},\tS)$. The emprical risk function is convex in its first argument (we are using a strictly convex regularizer) we have that:
\begin{equation}\label{eq:convex}
\hR(\hf, \tS) \leq \hR(f_{\alpha}, \tS) \leq \alpha \hR (\tf,\tS) + (1 - \alpha) \hR(\hf, \tS)
\end{equation}
However, since $\tf$ minimizes the loss with respect to $\tS$: $\hR(\tf, \tS) \leq \hR(\hf,\tS)$ and thus
\begin{equation}\label{eq:min-equal}
\hR(\hf, \tS) = \hR(\tf, \tS) = \min_{f \in \mathcal{F}} \hR(f,\tS)
\end{equation}
In other words we have shown that both $\hf$ and $\tf$ minimize the empirical risk with respect to $\tS$. The only thing that is left to be shown for the contradiction argument is that the tie breaking step of the algorithm does not distinguish between two functions that are risk minimizers. In order words, we need to show that both functions attain the minimum norm over all empirical risk minimizers.

Combining Equation~\eqref{eq:min-equal} with \eqref{eq:convex} we get that $\hR(f_{\alpha},\tS) \leq \hR(\hf,\tS)$. From \eqref{eq:risk} we have that $\hR(f_{\alpha},\tS) \leq \hR(\hf,\hS)$ and thus $\hR(f_{\alpha},\hS) = \hR(\hf,\hS)$. However, $\hf$ was chosen to miminimize the empirical risk with respect to $\hS$ and therefore, $||\hf||\leq || f_{\alpha}||$. Using convexity of the norm, we get $||\hf|| \leq ||\tf||$. Also, for the case of sample $\tS$, the algorithm chose function $\tf$ and therefore $||\hf|| \geq ||\tf||$. This concludes our contradiction argument, since 
\begin{equation}
||\hf|| = ||\tf|| = \min_{f \in \mathcal{F} : \hR(f,\tS) = \hR(\tf,\tS)}||f||
\end{equation}
Hence, both functions attain the minimum norm over all empirical risk minimizers. Since the norm is strictly convex, its minimum is unique and therefore $\hf \equiv \tf$.
\end{proof}

Using now the aforementioned proposition we will complete the proof of the theorem. Again, we follow the proof of \citet{DFP10}. Let $S = \{(\vx_i,y_i)\}_{i=1}^m$ be the set of the true reports of agents in $N$ and let $\tilde{S} = \{(\vx_i, \ty_i)\}_{i=1}^m$ be the reports revealed by the agents and used to traing the regression function. Let $C \subseteq N$ be an arbitrary coalition of agents that misreport their information, in order to decrease some of their respective losses. We define the hybrid set of values where $\forall i \in N$: $\hy_i = y_i$ if $i \in C$ and $\hy_i = \ty_i$ otherwise. Let $\hat{S} = \{(\vx_i, y_i)\}_{i=1}^m$, $\hf = \texttt{w-ERM-reg}(\calF, \ell, \hat{S})$ and $\tf = \texttt{w-ERM-reg}(\calF, \ell, \tilde{S})$.

If $\hf \equiv \tf$ then agents in $C$ have no incentive to misreport. If $\hf \neq \tf$ then from Proposition \ref{prop:dfp} we have that there exists an agent $i \in N$ such that $\hy_i \neq \ty_i$ and $\ell(\hf(\vx_i),\ty_i) < \ell(\tf(\vx_i),\hy_i)$. Since $\hy_i \neq \ty_i$, agent $i$ must be a member of $C$. Therefore, $\ty_i = y_i$ and $\ell(\hf(\vx_i),y_i) < \ell(\tf(\vx_i),y_i)$. However, no member of $C$ should lose from reporting $\tilde{S}$ instead of $\hat{S}$, contradiction. Since the proof holds regardless of the values revealed by the agents outside of $C$, we have group-strategyproofness.
\end{proof}

\subsection{CRM Mechanisms are Also GRL Mechanisms}

In the CRM mechanism, we refer to the point in $S$ which has the median of all median CWAs (i.e., DA) as the ``directing point'', and the point in $S'$ to which this DA is pointing as the ``directed point''.

\begin{proof}[Proof of Lemma~\ref{lem:crm-part-rl}]
	First, we show that for any $S \subseteq N$, the $(S,S)$-CRM mechanism is $(L,R,k,k')$-GRL mechanism for some $L,R,k,k'$. Without loss of generality, we can assume $S = N$ as the other points are simply ignored. Thus, we will refer to the $(N,N)$-CRM mechanism.
	
	First, consider the case where $n$ is even. Let $L$ (resp. $R$) be the set of $n/2$ points with the smallest (resp. largest) $x$ coordinates. We show equivalence of the $(N,N)$-CRM mechanism to the $(L,R,k,k')$-GRL mechanism for appropriate $k$ and $k'$. Let $(\beta_1,\beta_0)$ be the line returned by the CRM mechanism.
	
	Choose $x^* \in (\max_{i \in L} x_i, \min_{i \in R} x_i)$, and define the following sets.
	
	\begin{itemize}
		\item $A = \left\{i : x_i < x^*, y_i  \ge \beta_1 x_i + \beta_0\right\}$
		\item $B = \left\{i : x_i > x^*, y_i > \beta_1 x_i + \beta_0\right\}$
		\item $C = \left\{i : x_i < x^*, y_i  < \beta_1 x_i + \beta_0\right\}$
		\item $D = \left\{i  : x_i > x^*, y_i \le \beta_1 x_i + \beta_0\right\}$
	\end{itemize}
	
	Note that $A \cup C = L$ and $B \cup D = R$. For $i \in N$, let $MCWA_i$ denote the median CWA from $i$ to points in $N\setminus\set{i}$. Note that for each $i \in L$, there are strictly more points in $N\setminus\{i\}$ to the right of it, than to the left of it, implying that $MCWA_i \in [\pi,2\pi]$. Similarly, for each $i \in R$, we have $MCWA_i \in [0,\pi]$. 
	
	Let $DA$ be the directing angle under the CRM mechanism. Then, $DA = \min_{i \in L} MCWA_i$ or $DA = \max_{i \in R} MCWA_i$ based on whether the outer median in the directing angle definition uses the right median or the left median. Let us assume it uses the left median, so $DA = \max_{i \in R} MCWA_i$. The proof for the other case is symmetric. 
	
	We now show that in this case, $B = C = \emptyset$. This would imply that the mechanism is equivalent to $(L,R,|L|,1)$-GRL because every point in $L$ has a non-positive residual while every point in $R$ has a non-negative residual. 
	
	Suppose for contradiction that $B \neq \emptyset$. Take a point $i_B \in B$. Note that $MCWA_{i_B} \le \max_{i \in R} MCWA_i = DA$. Note that the directing point $i^*$ is on the regression line, and hence $i^* \in D$. Then, one can check that if $x_{i_B} < x_{i^*}$, then $x_{i_B}$ has strictly less number of points to which its angle is less than $MCWA_{i_B}$ than $x_{i^*}$ has to which its angle is less than $MCWA_{i^*} = DA$. In the case $x_{i_B} > x_{i^*}$, the same happens but for points with angle greater than MCWA. This is a contradiction because each point has exactly $(n-2)/2$ points with angle more or less than its MCWA. Hence, $B = \emptyset$. Using a symmetric argument, we can establish $C = \emptyset$, which completes the proof. 
	
	We now consider the case where $n$ is odd. In this case, let $L$ (resp. $R$) be the set of $(n-1)/2$ points with the smallest (resp. largest) $x$-coordinate, and let $i^*$ be the point with the median $x$-coordinate. Once again, we have that $MCWA_i \in [\pi,2\pi]$ for each $i \in L$, and $MCWA_i \in [0,\pi]$ for each $i \in R$. We add $i^*$ to $L$ if $MCWA_{i^*} \in [\pi,2\pi]$, and to $R$ otherwise. Suppose we add it to $R$, and let $R' = R \cup \set{i^*}$. Then using an argument similar to above, we can check that the CRM mechanism is equivalent to $(L,R',k,k')$ for appropriate $k,k'$.

	The case where $S \cap S' = \emptyset$ and $\min(|S|,|S'|)=1$ is much simpler. Again, without loss of generality, we can consider $S \cup S' = N$, and for simplicity, consider the case where $n$ is even and $|S| = 1$. The other cases are similar. Let $S = \set{i^*}$. Without loss of generality, suppose there are more points to the right of $i^*$ than to the left of it. Let $R$ be the set of points to the right of $i^*$, and $L$ be the set of points to the left of $i^*$.  Then, it is easy to see that when we take the median CWA from $i^*$ (say, the left median, i.e., the $(n/2-1)^\th$ smallest CWA), it will always be towards a point in $R$. Moreover, it will be the $(n/2-1-|S|)^\th$ smallest CWA towards points in $R$. However, CWAs towards points in $R$ are monotonic in slopes to points in $R$. Hence, the regression line will make the $(n/2-1-|S|)^\th$ smallest residual in $R$ zero. In other words, the mechanism is equivalent to $(\set{i^*},R,1,n/2-1-|S|)$-GRL.
\end{proof}

\subsection{Impartial Mechanisms}

We now present the proof of Theorem~\ref{thm:impartial}. First, we need the following definition. 

\begin{definition}[Completely Additively Separable]
	Function $f: \bbR^k \to \bbR$ is called {\em completely additively separable} if there exist functions $\set{g_i}_{i=1}^k$ such that $f(t_1,\ldots,t_k) = \sum_{i=1}^k g_i(t_i)$ for all $\vec{t} = (t_1,\ldots,t_k) \in \bbR^k$. 
\end{definition}

It is well known that $f$ is completely additively separable if and only if for all $\vec{t} \in \bbR^k$, $i \in [k]$, and $t'_i \in \bbR$, $f(t_i,\vec{t}_{-i})-f(t'_i,\vec{t}_{-i})$ is independent of $\vec{t}_{-i}$. 

\begin{proof}[Proof of Theorem~\ref{thm:impartial}]
	We omit $\vx$ from all superscripts for simplicity. Suppose mechanism $M$ is given by Equation~\eqref{eq:impartial-hyperplane}. Then:
	\begin{align*}
	\hy_i(\vbeta) = \lbr \vbeta, \bvx(i) \rbr &= \lbr \textstyle\sum_{j \in N} g_j(y_j), \vx(i) \rbr + c - \textstyle\sum_{j \in N} \lbr g_j(y_j), \vx(j) \rbr\\
	&= c + \textstyle\sum_{j \in N\setminus\set{i}} \lbr g_j(y_j), \vx(i)-\vx(j) \rbr.
	\end{align*}
	Note that $\hy_i(\vbeta)$ is independent of $y_i$, which implies that $M$ is impartial.
	
	We now prove the converse for simple linear regression ($d=1$) with an admissible set of points. Suppose mechanism $M$ is impartial. Given $\vy$, let $\beta_1(\vy)$ be the slope of the line returned by $M$, and $f_i(\vy) = \hy_i(M(\vy))$ be the outcome for agent $i$. Because $M$ is impartial, $f_i$ is independent of $y_i$. Hence, we denote the outcome for agent $i$ by $f_i(\vy_{-i})$. 
	
	We want to show that $h$ is completely additively separable. Equivalently, for every $\vy$ and $\tvy$ such that $\vy_{-i} = \tvy_{-i}$, we want to show that $\beta_1(\vy)-\beta_1(\tvy)$ is independent of $\vy_{-i}$. Choose $j \in N\setminus\set{i}$ arbitrarily. By the definition of the slope of a line, we have
	$$
	\beta_1(\vy) = \frac{f_j(\vy_{-j})-f_i(\vy_{-i})}{x_j-x_i},\ \beta_1(\tvy) = \frac{f_j(\tvy_{-j})-f_i(\tvy_{-i})}{x_j-x_i}.
	$$
	Taking the difference, and noting that $\vy_{-i} = \tvy_{-i}$, we get
	$$
	\beta_1(\vy)-\beta_1(\tvy) = \frac{f_j(\vy_{-j})-f_j(\tvy_{-j})}{x_j-x_i}.
	$$
	Note that the RHS is independent of $y_j$. Since we chose $j \in N\setminus\set{i}$ arbitrarily, it follows that $\beta_1(\vy)-\beta_1(\tvy)$ is independent of $\vy_{-i}$, implying that $h$ is completely additively separable. Thus, there must exist functions $\set{g_i}_{i \in N}$ such that $\beta_1(\vy) = \sum_{i \in N} g_i(\vy)$. 
	
	We now want to calculate $\beta_0$. Recall that for every $i \in N$, the outcome for agent $i$ is
	$$
	f_i(\vy_{-i}) = \beta_1(\vy) \cdot x_i + \beta_0 = g_i(y_i) \cdot x_i + \sum_{j \in N\setminus\set{i}} g_j(y_j) \cdot x_i + \beta_0.
	$$
	Since the LHS is independent of $y_i$, so must be the RHS. Hence, $\beta_0 + g_i(y_i)\cdot x_i$ must be independent of $y_i$ for each $i \in N$. This implies $\beta_0 = c - \sum_{i \in N} g_i(y_i)\cdot x_i$ for some constant $c$, as desired.
\end{proof}

\begin{proof}[Proof of Proposition~\ref{prop:impartial-not-gsp}]
	By Theorem~\ref{thm:impartial}, an impartial mechanism for simple linear regression with an admissible set of points must be of the form given in Equation~\eqref{eq:impartial-hyperplane}. We want to show that function $g_i^{\vx}$ is constant for each $i \in N$. Suppose for contradiction that for some agent $i \in N$, function $g_i^{\vx}$ is not constant. Thus, there exist $y^1_i$ and $y^2_i$ such that $g_i^{\vx}(y^1_i) \neq g_i^{\vx}(y^2_i)$. Fix an agent $j \in N\setminus\set{i}$ and $\vy_{-\set{i,j}} \in \bbR^{n-2}$. Let $\hat{y^1_j}$ and $\hat{y^2_j}$ denote the outcomes for agent $j$ under the impartial mechanism when agent $i$ reports $y^1_i$ and $y^2_i$, respectively, and agents in $N\setminus\set{i,j}$ report $\vy_{-\set{i,j}}$. That is, 
	$$
	\hat{y^t_j} = g_i^{\vx}(y^t_i) \cdot (x_j-x_i) + \textstyle\sum_{k \in N\setminus\set{i,j}} g_k^{\vx}(y_k) \cdot (x_j-x_k) + c^{\vx}, \forall t \in \set{1,2}.
	$$
	
	Note that $g_i^{\vx}(y^1_i) \neq g_i^{\vx}(y^2_i)$ and $x_i \neq x_j$ imply that $\hat{y^1_j} \neq \hat{y^2_j}$. Now, suppose that the private values of the agents are $(y^1_i,\hat{y^2_j},\vy_{-\set{i,j}})$. In this case, the outcome for agent $j$ is $\hat{y^1_j}$, which is different from her private value $\hat{y^2_j}$. If agent $i$ changes her report to $y^2_i$, her own outcome would not change, but the outcome for agent $j$ would change to $\hat{y^2_j}$, making agent $j$ strictly better off. Thus, the coalition $\set{i,j}$ successfully manipulates their reports, showing a violation of group strategyproofness.
	
	For the reverse direction, note that all constant functions are trivially group strategyproof.
\end{proof}

\subsection{Characterization of Strategyproof Mechanisms}

\begin{proof}[Proof of Lemma~\ref{lem:1dim-1agent}]
	Part~\ref{thm1agent-parta} is precisely the characterization of strategyproof mechanisms due to~\citet[Proposition 3]{Moul80}, applied to the case of a single agent.\footnote{Equivalently, one can use Proposition~2, which characterizes strategyproof and anonymous mechanisms, as anonymity becomes trivial in case of a single agent.}
	
	We would like to show that part~\ref{thm1agent-partb} is equivalent to part~\ref{thm1agent-parta}. It is easy to check that a function $\pi$ of the form given in part~\ref{thm1agent-parta} satisfies the conditions of part~\ref{thm1agent-partb}. We now show the converse. 
	
	Suppose that $\pi$ is continuous, and for every $y \in \bbR$, either $\pi(y) = y$ or $\pi$ is locally constant at $y$. Let $O = \set{y \in \bbR : \pi \text{ is locally constant at } y}$. We first show that $O$ is an open set. That is, if $y \in O$, there must exist a $\delta > 0$ such that $(y-\delta,y+\delta) \subseteq O$. Indeed, fix a $y \in O$. Because $\pi$ is locally constant at $y$, there must exist an $\eps > 0$ such that $\pi$ is constant in $[y-\eps,y+\eps]$. Set $\delta = \eps/2$, and pick an arbitrary $y' \in (y-\delta,y+\delta)$. We want to show that $y' \in O$. Note that for $\eps' = \eps/2$, $[y'-\eps',y'+\eps'] \subseteq [y-\eps,y+\eps]$. Hence, $\pi$ is constant in $[y'-\eps',y'+\eps']$, implying that $y' \in O$. This concludes the proof that $O$ is an open set.
	
	Next, we use the well-known fact that any open subset of $\bbR$ is a countable union of pairwise disjoint open intervals. That is, we can write $O = \cup_{k \in \bbN}\ (a_k,b_k)$, where $a_k,b_k \in \bbRbar$. For $k \in \bbN$, because $\pi$ is locally constant over $(a_k,b_k)$, and an open interval is a connected metric space, it follows that $\pi$ is globally constant over $(a_k,b_k)$. That is, there exists a value $t_k \in \bbR$ such that $\pi(y) = t_k$ for all $y \in (a_k,b_k)$. 
	
	We now show that for any $k \in \bbN$ with $a_k \neq b_k$ (i.e., the interval $(a_k,b_k)$ is non-empty), it cannot be the case that both $a_k$ and $b_k$ are finite. Suppose for contradiction that both are finite. Note that continuity of $\pi$ implies that $\pi(a_k) = \pi(b_k) = t_k$. However, since $a_k,b_k \notin O$, we have $\pi(a_k) = a_k$ while $\pi(b_k) = b_k$, which is a contradiction because $a_k \neq b_k$. Hence, for every $k \in \bbN$ with $a_k \neq b_k$, at least one of the two must lie in $\set{-\infty,\infty}$. 
	
	This leaves precisely five possibilities for the set $O$: $\emptyset$, $\bbR$, $(-\infty,a)$ for $a  \in \bbR$, $(b,\infty)$ for $b \in \bbR$, and $(-\infty,a) \cup (b,\infty)$ for $a,b \in \bbR$ with $b \ge a$. We know that $\pi$ is constant over each interval in $O$, and the identity function for every point outside $O$. For each of these five cases, we show that $\pi$ must be of the form given in part~\ref{thm1agent-parta} by identifying the corresponding constants $\alpha^1$ and $\alpha^2$.
	\begin{enumerate}
		\item $O = \emptyset$: $\pi$ is the identity function everywhere, i.e., $\alpha^1 = -\infty$ and $\alpha^2 = \infty$.
		\item $O = \bbR$: There exists $t \in \bbR$ such that $\pi(y) = t$ for all $y \in \bbR$. This corresponds to $\alpha^1 = \alpha^2 = t$. 
		\item $O = (-\infty,a)$ for $a \in \bbR$: Then $\pi(y) = y$ for all $y \ge a$. In particular, $\pi(a) = a$. Because $\pi$ is continuous and constant over $(-\infty,a)$, we have $\pi(y) = a$ for $y \in (-\infty,a)$. This corresponds to $\alpha^1 = a$ and $\alpha^2 = \infty$. 
		\item $O = (b,\infty)$ for $b \in \bbR$: Similarly to case (3), this corresponds to $\alpha^1 = -\infty$ and $\alpha^2 = b$. 
		\item $O = (-\infty,a) \cup (b,\infty)$ for finite $b \ge a$: As argued in the previous two cases, for $y \in (-\infty,a)$ we have $\pi(y) = \pi(a) = a$, and for $y \in (b,\infty)$ we have $\pi(y) = \pi(b) = b$. For $y \in [a,b]$, we have $\pi(y) = y$. This corresponds to $\alpha^1 = a$ and $\alpha^2 = b$. 
	\end{enumerate}
	This concludes our proof.
\end{proof}

\subsection{Efficiency of Strategyproof Mechanisms}

Figure~\ref{fig:mathematica} below verifies several claims made in the proof of Theorem~\ref{thm:efficiency-lower-bound} using Mathematica.

\begin{figure}[!ht]
	\centering
	\includegraphics[width=\textwidth]{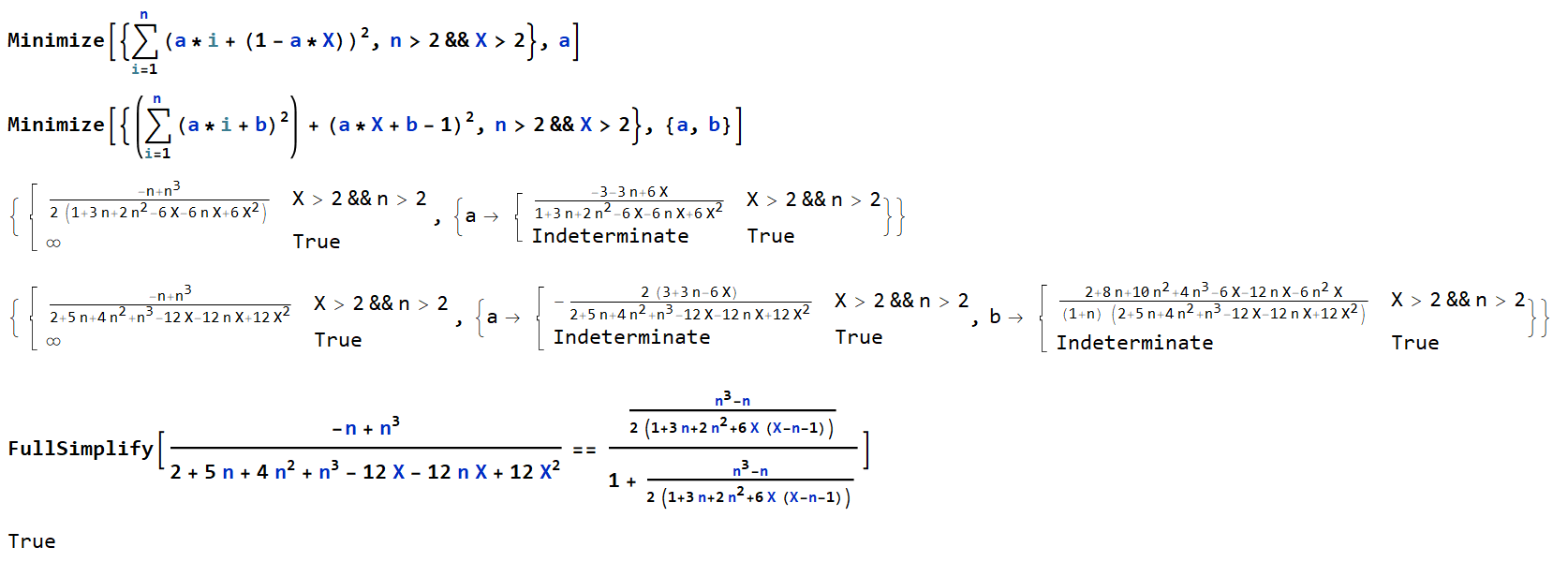}
	\caption{Verification of various claims through Mathematica}
	\label{fig:mathematica}
\end{figure}

We remark that none of the strategyproof mechanisms we study achieve a constant approximation. For instance, it is easy to show that $\erm$ is $n$-efficient.

\begin{proposition}\label{prop:erm-efficient}
	The $\erm$ mechanism is $n$-efficient.
\end{proposition}
\begin{proof}
	Fix $\calD = (\vx(i),y_i)_{i \in N}$. Let $\vec{\beta^1}$ and $\vec{\beta^*}$ be the outputs of $\erm$ and OLS, respectively. Then, we have
	$$
	\textstyle\rss(\calD,\vec{\beta^1}) \le \left( \sum_{i \in N} \left|y_i-(\vec{\beta^1})^T\ \bvx(i)\right|  \right)^2 \le \left( \sum_{i \in N} \left|y_i-(\vec{\beta^*})^T\ \bvx(i)\right|  \right)^2 \le n \cdot \rss(\calD,\vec{\beta^*}),
	$$
	where the first inequality follows from the power mean inequality, the second inequality holds because $\vec{\beta^1}$ minimizes the sum of absolute losses, and the third inequality follows from the Cauchy-Schwarz inequality. This concludes the proof.
\end{proof}

\section{Computing Influence Bounds}
\label{appendix:influence-bounds}

\begin{algorithm}
	\SetAlgoLined
	\KwIn{Data points $\calD = (\vx(j),y_j)_{j \in N}$, agent $i \in N$.}
	\KwOut{$\ell_i, h_i$}
	$Z \gets$ set of hyperplanes $\vbeta$ which pass through $d+1$ agents from $N\setminus\set{i}$\;
	$t_{\vbeta} \gets \vbeta^T\ \bvx(i), \forall \vbeta \in Z$\;
	$L \gets \min_{\vbeta \in Z} t_{\vbeta}-1$\;
	$H \gets \max_{\vbeta \in Z} t_{\vbeta}+1$\;
	$V_L \gets M^{\vx}(L,\vy_{-i})$\;
	$V_H \gets M^{\vx}(H,\vy_{-i})$\;
	\uIf {$V_L = L$}{
		$\ell_i \gets - \infty$\;
	}
	\Else {
		$\ell_i \gets V_L$\;
	} 
	\uIf {$V_H = H$}{
		$h_i \gets \infty$\;
	}
	\Else {
		$h_i \gets V_H$\;
	}
	\Return $\ell_i,h_i$\;
	\caption{Computing Influence Bounds}
	\label{alg:infl-bounds}
\end{algorithm}

Our characterization result (Theorem~\ref{thm:our-char}) establishes existence of influence bounds $\ell_i,h_i \in \bbRbar$ for each agent $i$ as a function of the reports of the other agents. In this section, we address the problem of computing these influence bounds for a given strategyproof mechanism.

Fix $\vy_{-i}$. We begin from the simple observation that if $\ell_i$ is finite, then for a sufficiently low value of $y_i$ (any $y_i \le \ell_i$), we have that the outcome for agent $i$ will be $\hy_i = \med(y_i,\ell_i,h_i) = \ell_i$. If $\ell_i = -\infty$, then for all $y_i < h_i$, the outcome for agent $i$ will be $\hy_i = y_i$. Thus, if we can identify a {\em sufficiently low} value of $y_i$, we can check if $\hy_i$ is equal to $y_i$ (in which case $\ell_i = -\infty$), or $\hy_i$ is equal to some other value (in which case this value must be $\ell_i$). A symmetric observation holds for $h_i$. 

While it is difficult to pin down a sufficiently low value for an arbitrary strategyproof mechanism, we can do so for the class of strategyproof mechanisms which are guaranteed to pass through $d+1$ data points in $d+1$ dimensions (e.g., the generalized resistant hyperplane mechanisms). 

In this case, note that $\ell_i$, if finite, must be the point where a hyperplane containing {\em some} $d+1$ agents (excluding agent $i$) intersects the vertical line at $\vx(i)$. Thus, if we iterate through all hyperplanes passing through $d+1$ agents except agent $i$, and find their intersections with the vertical line at $\vx(i)$, then any value lower than the lowest intersection point will work as a sufficiently low value. Once again, a symmetric observation can be made for $h_i$. 

This provides an algorithm that runs in time that is polynomial in $n$, but exponential in $d$, and makes two calls to the strategyproof mechanism (one to identify $\ell_i$ and one for $h_i$). This is presented as Algorithm~\ref{alg:infl-bounds}.

\section{Quantile Regression is Not Strategyproof}
\label{appendix:qERM-cnt}

In this section, we show that quantile regression is not guaranteed to be strategyproof. In particular, we show that quantile regression with $q = 0.4$ violates strategyproofness. The coordinates for the $20$ data points shown in Figure~\ref{fig:cnt-qERM} are as follows.

\begin{center}
\begin{tabular}{c c c c c}
(-79.3, -45.8) & (-77.3, 89.5) & (-74.8, -87.4) & (-58.5, 14.3) &  (-33.2, -28.4) \\
(-31.5, 5.2) & (-8.0, -73.1) & (-1.7, -52.8) & (10.0, 88.6) &  (13.0, 13.3) \\
(13.9, 7.4) & (15.4, 39.4) & (18.5, -2.0) & (23.0, 6.6) &  (23.8, -33.0) \\ 
(24.2, -60.3) & (26.0, 49.5) & (39.5, 49.5) & (45.3, 88.9) &  (71.2, 33.2)
\end{tabular}
\end{center}

If the agents report truthfully, then the quantile regression mechanism with $q=0.4$ returns the solid line. If agent with data point $(13.9, 7.4)$ reports a very large value of $y$ (e.g., $2000$), then the output line becomes the dashed one, which is clearly beneficial for the manipulating agent.

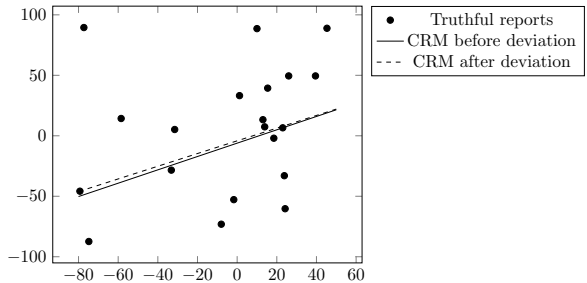
\begin{figure}[h]
	\centering
\begin{tikzpicture}[scale=0.6]
\begin{axis}[
    legend pos=outer north east
]
\addplot [only marks] table {
-79.3 -45.8
-77.3 89.5
-74.8 -87.4
-58.5 14.3
-33.2 -28.4
-31.5 5.2 
-8.0 -73.1 
-1.7 -52.8 
10.0 88.6 
13.0 13.3 
13.9 7.4 
15.4 39.4 
18.5 -2.0 
23.0 6.6 
23.8 -33.0 
24.2 -60.3 
26.0 49.5 
39.5 49.5 
45.3 88.9
1.2 33.2
};
\addlegendentry{Truthful reports}

\addplot [domain=-80:50, samples=2] {0.5518*x -6.0929};
\addlegendentry{CRM before deviation}

\addplot [domain=-80:50, samples=2,dashed] {0.5249*x - 4.1742};
\addlegendentry{CRM after deviation}

\end{axis}
\end{tikzpicture}
	\caption{Example of a beneficial manipulation under quantile regression with $q = 0.4$.}
    \label{fig:cnt-qERM}
\end{figure}

\end{document}